\documentclass[letterpaper,twocolumn,accepted=2021-05-21]{quantumarticle}
\pdfoutput=1
\usepackage{amsmath,amssymb,graphicx,mathtools}
\usepackage{hhline,bbm}
\usepackage{float}

\usepackage[T1]{fontenc}
\usepackage[numbers,sort&compress]{natbib}

\makeatletter
\DeclareRobustCommand{\cev}[1]{%
  \mathpalette\do@cev{#1}%
}
\newcommand{\do@cev}[2]{%
  \fix@cev{#1}{+}%
  \reflectbox{$\m@th#1\vec{\reflectbox{$\fix@cev{#1}{-}\m@th#1#2\fix@cev{#1}{+}$}}$}%
  \fix@cev{#1}{-}%
}
\newcommand{\fix@cev}[2]{%
  \ifx#1\displaystyle
  \mkern#23mu
  \else
  \ifx#1\textstyle
  \mkern#23mu
  \else
  \ifx#1\scriptstyle
  \mkern#22mu
  \else
  \mkern#22mu
  \fi
  \fi
  \fi
}

\makeatother

\newcounter{theoremcounter}
\stepcounter{theoremcounter}

\newtheorem{theorem}{Theorem}
\newtheorem{lemma}{Lemma}
\newtheorem{proposition}{Proposition}
\newtheorem{corollary}{Corollary}

\newenvironment{proof}[1][Proof]{\begin{trivlist}
\item[\hskip \labelsep {\bfseries #1}]}{\end{trivlist}}

\newcommand{\qed}{\hfill $\blacksquare$}

\newcommand{\bs}{\boldsymbol}
\newcommand{\bra}[1]{\left\langle #1\right|}
\newcommand{\ket}[1]{\left|#1\right\rangle}
\newcommand{\braket}[2]{\left\langle #1|#2\right\rangle}

\newcommand{\ketbra}[2]{\ket{#1}\bra{#2}}
\newcommand{\partder}[2]{\frac{\partial #1}{\partial #2}}

\DeclareMathOperator{\Tr}{Tr}

\newcommand{\Mod}[1]{\ \mathrm{mod}\ #1}
\newcommand{\bsmc}[1]{\bs{\mathcal #1}}

\makeatletter
\newcommand{\vast}{\bBigg@{4}}
\newcommand{\Vast}{\bBigg@{5}}
\makeatother

\begin{document}

\title{Stationary Phase Method in Discrete Wigner Functions and Classical Simulation of Quantum Circuits}
\author{Lucas Kocia}
\affiliation{National Institute of Standards and Technology, Gaithersburg, Maryland, 20899, U.S.A.}
\author{Peter Love}
\affiliation{Department of Physics, Tufts University, Medford, Massachusetts 02155, U.S.A.}
\begin{abstract}
  One of the lowest-order corrections to Gaussian quantum mechanics in infinite-dimensional Hilbert spaces are Airy functions: a uniformization of the stationary phase method applied in the path integral perspective. We introduce a ``periodized stationary phase method'' to discrete Wigner functions of systems with odd prime dimension and show that the \(\frac{\pi}{8}\) gate is the discrete analog of the Airy function. We then establish a relationship between the stabilizer rank of states and the number of quadratic Gauss sums necessary in the periodized stationary phase method. This allows us to develop a classical strong simulation of a single qutrit marginal on \(t\) qutrit \(\frac{\pi}{8}\) gates that are followed by Clifford evolution, and show that this only requires \(3^{\frac{t}{2}+1}\) quadratic Gauss sums. This outperforms the best alternative qutrit algorithm (based on Wigner negativity and scaling as \(\sim\hspace{-3pt} 3^{0.8 t}\) for \(10^{-2}\) precision) for any number of \(\frac{\pi}{8}\) gates to full precision.
\end{abstract}
\maketitle

\section{Introduction}
\label{sec:intro}

This paper aims to bridge the gap between efficient strong simulation qubit methods, which use low stabilizer rank representations of magic states, and odd-prime-dimensional qudit methods, which have relied on Wigner negativity in past studies. We introduce a strong simulation odd-dimensional qudit algorithm that is more efficient than the current state-of-the-art algorithm based on Wigner negativity.

A \emph{strong} simulation refers to the computation of a quantum probability. This differs from \emph{weak} simulation, which refers to the computation of sample outcomes of a quantum probability distribution. Strong simulation of quantum universal circuits up to multiplicative error is \(\#P\)-hard~\cite{Kuperberg15,Fujii17,Hangleiter18}; the cost of classical strong simulation is expected to not only scale exponentially but to also be harder than tasks that can be performed efficiently on a quantum computer. However, this does not mean that there do not exist efficient methods in particular cases, such as those frequently covered in introductory textbooks on quantum mechanics.

For instance, by representing quantum circuits through the path integral perspective, efficient approximation schemes can be formed by only including lowest-order terms in \(\hbar\). It is known that quantum circuits require terms at orders greater than \(\hbar^0\) for them to attain quantum universality~\cite{Kocia17,Kocia17_2,Kocia17_3}. In the continuous-variable setting, this corresponds to requiring non-Gaussianity as a resource. The ``minimal'' non-Gaussianity that can be added takes the functional form of Airy functions, and these can be treated within the path integral formalism with just one more order of correction in terms of \(\hbar\): \(\mathcal O(\hbar^1)\)~\cite{Heller18}. Perhaps a similar systematic expansion in orders of \(\hbar\) is also possible in the discrete case, and can result in a more efficient treatment of some universal gates.

In the discrete case, non-Gaussianity has been studied in other ways not explicitly dependent on orders of \(\hbar\). Such past work on efficient classical simulation has utilized the negativity of Wigner functions or contextuality~\cite{Spekkens08,Delfosse15} as a resource of non-Gaussianity. More formally, the resource theory of contextuality in finite-dimensional systems has also been explicitly considered recently~\cite{Abramsky17,Duarte17}.

For quantum states \(\hat \rho\), the Wigner function~\cite{Wigner32,Wootters87} \(\rho(\bs x)\) is a non-negative function for stabilizer states. Since Clifford gates take stabilizer states to stabilizer states, classical simulation of stabilizer states under Clifford gate evolution using this discrete Wigner formulation can be accomplished in polynomial time and has been shown to be non-contextual~\cite{Kocia17}.

However, Clifford gates and stabilizer states do not allow for universal quantum computation. The extension of this set by \(\frac{\pi}{8}\) gates or magic states to allow for universal quantum computation introduces contextuality into \(\rho(\bs x)\), which means that they can also take negative values~\cite{Ferrie09,Mari12,Veitch12,Veitch13}. Since the contextual part of a classical simulation is solely responsible for its transition from polynomial to exponential scaling, it is useful to treat contextuality as an operational resource~\cite{Chitambar18} that can be added to a polynomially-efficient backbone in the appropriate amount to allow for a quantum process to be simulated. The goal of such a contextuality injection scheme is a classical computation cost that scales as efficiently as possible with the contextuality present.

Previous attempts at developing a framework for classical algorithms for simulation of qudit quantum circuits that leverage contextuality have done so indirectly by using Wigner function negativity~\cite{Pashayan15}. Recently, there have also been some more direct approaches~\cite{Duarte17}. In continuous systems, a historically successful approach of leveraging contextuality efficiently is to use semiclassical techniques that rely on expansion in orders of \(\hbar\). The relationship between contextuality and higher orders of \(\hbar\) in the Wigner-Weyl-Moyal (WWM) formalism has been recently established~\cite{Kocia17,Kocia17_2,Kocia17_3}. Previous derivation of WWM in odd dimensions~\cite{Almeida98,Rivas99,Kocia16} were not able to accomplish this because they began by using results from the continuous case of the WWM formalism. In the continuous regime, the stationary phase method is applies in its traditional form, and must then be ``discretized'' and ``periodized''. Following that approach does not allow for the derivation of higher order \(\hbar\) corrections because anharmonic trajectories cannot be obtained by ``periodizing'' to a discrete Weyl phase space grid~\cite{Kocia16}.

Here we are able to develop a method that includes such higher order terms for discrete odd prime-dimensional systems. Without appealing to the continuous case, we use a periodized version of Taylor's theorem to develop the stationary phase method in the discrete setting. We explain how this ``periodized'' form of the stationary phase method differs from the continuous version of the stationary phase method.

This stationary phase method allows us to derive the WWM formalism with asymptotically decreasing order \(\hbar\) corrections to a subset of unitaries that include the \(T\)-gate and \(\frac{\pi}{8}\) gates. These serve as contextual corrections to non-contextual Clifford operations. We obtain a new method for classical simulation of quantum circuits involving Clifford+\(\frac{\pi}{8}\) gates for qutrits. To demonstrate the effectiveness of our resultant method, we calculate its computational cost compared to a recent study that introduced a method to sample the same odd-dimensional Wigner distribution using Monte Carlo methods~\cite{Pashayan15}. The stationary phase method scales as \(3^{\frac{t}{2}+1}\) for calculations to full precision while the best alternative qutrit algorithm scales as \(\sim\hspace{-3pt} 3^{0.8 t}\) for calculations to \(10^{-2}\) precision. This allows the stationary phase method to be more useful for intermediate-sized circuits

Related approaches include recent proposals to use non-Gaussianity as a resource in continuous (infinite dimensional) quantum mechanics such as optics~\cite{Zhuang18,Takagi18,Albarelli18}, as well as the well-established field dealing with semiclassical propagators in continuous systems~\cite{Tannor07,Brack18,Heller18}.

This paper is organized as follows: in Section~\ref{sec:statphaseincontsys} we review the stationary phase method in continuous systems. In Section~\ref{sec:WWMformalism} we briefly introduce aspects of the WWM formalism and develop its quantum channel representation of gates. We introduce closed-form results for quadratic Gauss sums in Section~\ref{sec:Gausssums} and this motivates the introduction of the periodized stationary phase method for discrete odd prime-dimensional systems in Section~\ref{sec:statphase}. Section~\ref{sec:uniformization} compares and characterizes this discrete stationary phase method with its continuous analog and uniformization. The periodized stationary phase method is then used to evaluate the qutrit \(\pi/8\) gate magic state in Section~\ref{sec:qutritpi8gate}. We then establish a relationship between stabilizer rank, Wigner phase space dimension and number of critical points in Section~\ref{sec:stabrank}. These results allow us to produce an expression with \(3^{\frac{t}{2}}\) critical points for \(t\) states in Section~\ref{sec:twoqutritmagicstates}. We discuss future directions of study in Section~\ref{sec:future} and close the paper with some concluding remarks in Section~\ref{sec:conc}.

\section{Stationary Phase in Continuous Systems}
\label{sec:statphaseincontsys}

The stationary phase method is an exact method of reexpressing the integral of an exponential of a function in terms of a sum of contributions from critical points:
\begin{equation}
  \label{eq:semiclassprop}
\int \mathcal{D}[x] e^{\frac{i}{\hbar} S[x]} = \sum_j \int \mathcal{D}[x_{j}] \, e^{\frac{i}{\hbar} \left( S[x_{j}] + \delta S[x_{j}] + \frac{1}{2} \delta^2 S[x_{j}] + \ldots \right) },
\end{equation}
where \(S[x_j]\) is a functional over \(x_j\), and the sum is indexed by critical ``trajectories'' \(j\) defined by
\begin{equation}
  \label{eq:statphasecond}
  \left.\frac{\delta S[x]}{\delta x}\right|_{x=x_j} = 0,
\end{equation}
and we have chosen to exercise our freedom to factor out the term \(\frac{i}{\hbar}\) from \(S\), for later clarity. These contributions from critical points correspond to Gaussian integrals, or Fresnel integrals in the general case, which can be analytically evaluated, or higher-order uniformizations corresponding to two or more critical points. In general, the higher order contributions asymptotically decrease in significance. When \(S\) is the classical action of a particle, Eq.~\ref{eq:statphasecond} defines the critical points as classical trajectories that satisfy the Hamiltonian associated with \(S\). (In this case, \(S\) is functional of \(x\) because it is the integral of the Lagrangian over time, the latter of which is a function of \(x\) and \(\dot x\)).

Terminating Eq.~\ref{eq:semiclassprop}'s expansion of \(S[x]\) at second order corresponds to making a first order approximation in \(\hbar\):
\begin{eqnarray}
  \label{eq:vVMG}
\int \mathcal{D}[x] e^{\frac{i}{\hbar} S[x]} = \sum_j\left( \frac{- \frac{\partial^2 S_{j}}{\partial \bs x \partial \bs x'}}{2 \pi i \hbar} \right)^{1/2} e^{i \frac{S_j(x,x')}{\hbar}} + \mathcal{O}(\hbar^2).\nonumber
\end{eqnarray}
where the sum is over all classical paths that satisfy the boundary conditions \(\partder{S}{x} = \partder{S}{x'} = 0\). The subscript \(j\) in \(S_j(\bs x,\bs x')\) indexes  the generally infinite number of classical trajectories that satisfy the initial \(x\) and final \(x'\) conditions (i.e. \(x\) and \(x'\) are not generally sufficient to define the functional \(S[x(t)] \equiv S(x,x')\)).

Here we will be interested in the double-ended propagator, or quantum channel---the amplitude from one state to another state. Thus, we will discuss the stationary phase approximation in reference to two integrals. We will also be interested in the case that the action \(S\) can be expressed as a polynomial in its arguments and thus becomes a \emph{bona fide} function. Given functions \(\psi(x'')\) and \(\psi'(x''')\) that represent these states, and an action \(S\) that represents the system of interest (usually this is an integral of the system's Lagrangian and so is related to its Hamiltonian), the propagator between them can be written as
\begin{eqnarray}
  &&U(\psi, \psi', S) \nonumber\\
  &=& \int^{\infty}_{-\infty} \text d \bs x'' \int^{\infty}_{-\infty} \text d \bs x''' \psi^*(\bs x'') U(\bs x'', \bs x''', t) \psi'(\bs x''')\\
  &=& \int^{\infty}_{-\infty} \text d \bs x'' \int^{\infty}_{-\infty} \text d \bs x''' \psi^*(\bs x'') e^{i S(\bs x'', \bs x'''; t)/\hbar} \psi'(\bs x''')\\
  &=& \sum_j \left( S^2_j \right)^{1/2} e^{i S_j /\hbar} + \mathcal O (\hbar^2)\nonumber
\end{eqnarray}
where the action \(S\) is assumed to be a polynomial in \(\mathbb C[\bs x''_1, \ldots, \bs x''_n, \bs x'''_1, \ldots, \bs x'''_n]\) (a complex-valued function of \(\bs x''_i\) and \(\bs x'''_n\)), \(S_j\) corresponds to the action evaluated at the \(j\)th zero of \(\partder{S}{(x'',x')}\)---the \(j\)th ``critical point'' of \(S\)---and \(S^2_j\) is a function of the second derivative of \(S\) evaluated at the \(j\)th critical point.

As an example, if one chooses to represent the propagator in terms of initial and final position, \(\bra \psi = \bra q\) and \(\ket \psi' = \ket q'\) respectively, then the above equation becomes
\begin{eqnarray}
  &&U(\bs q, \bs q', S) \nonumber\\
  &=& \int^{\infty}_{-\infty} \text d \bs q'' \int^{\infty}_{-\infty} \text d \bs q''' \braket{q}{\bs q''} e^{i S(\bs q'', \bs q'''; t)/\hbar} \braket {\bs q'''}{q'}\\
  &=& \sum_j \left( \frac{\partial^2 S}{\partial q_j \partial q'_j} \right)^{1/2} e^{i S(\bs q_j, \bs q'_j; t)/\hbar} + \mathcal O (\hbar^2)\nonumber,
\end{eqnarray}
where the sum is over the critical points \(\bs q_j\) and \(\bs q'_j\) of the polynomial action \(S\) function. The first term is often called the \emph{Van Vleck propagator} or, more precisely, the \emph{Van Vleck-Morette-Gutzwiller propagator}~\cite{Van28,Morette51,Gutzwiller67}.

However, we have purposely been general with our choice of representation \(\psi\) and \(\psi'\) because instead of the position representation, we will employ the ``center-chord'' representation. This ``center-chord'' representation is more suitable for discrete Hilbert spaces~\cite{Kocia16} and produces the Wigner-Weyl-Moyal (WWM) formalism. Using this representation with the stationary phase method means that stationary phase expansions must be made around the \emph{centers} \(\bs x\) and \(\bs x'\). We briefly develop aspects of this formalism that are pertinent to gate concatenation in the next section.

\section{Discrete Wigner Functions: Wiger-Weyl-Moyal Formalism}
\label{sec:WWMformalism}

The discrete WWM formalism of quantum mechanics is equivalent to the matrix and path integral representations~\cite{Berezin77,Rivas99}. A complete description of the odd-dimensional WWM formalism can be found elsewhere~\cite{Almeida98,Rivas99,Kocia16}. Here we introduce the basic formalism of Weyl symbols of products of operators or gates, which will be necessary in order to develop a propagator-like treatment.

\subsection{WWM Formalism Basics}

We consider the case of \(d\) odd. We will later restrict this further to \(d = p^h\) for \(p\) odd prime and \(h > 0\), not to be confused with the reduced Planck's constant. We set \(\hbar = \frac{d}{2 \pi}\).

We label the computational basis for our system by \(j \in {0, 1, \ldots, d-1}\), for \(d\) odd. We identify the discrete position basis with the computational basis and define the ``boost'' operator as diagonal in this basis:
\begin{equation}
  \hat {Z}^{\delta p} \ket{j} \equiv \omega^{j \delta p} \ket{j},
\end{equation}
where \(\omega \equiv \omega(d) = e^{2 \pi i/d}\).

We define the normalized discrete Fourier transform operator to be equivalent to the Hadamard gate:
\begin{equation}
  \hat {F} = \frac{1}{\sqrt{d}}\sum_{j,k \in \mathbb{Z}/d\mathbb{Z}} \omega^{-j k} \ket{j}\bra{k}.
  \label{eq:hadamard}
\end{equation}
This allows us to define the Fourier transform of \(\hat {Z}\) as the ``shift'' operator:
\begin{equation}
   \hat {X} \equiv \hat {F} \hat {Z} \hat {F}^\dagger,
\end{equation}
which acts as
\begin{equation}
  \hat {X}^{\delta q} \ket{j} \equiv \ket{j\oplus\delta q},
\end{equation}
where \(\oplus\) denotes mod-\(d\) integer addition.
The Weyl relation holds for \(\hat X\) and \(\hat Z\):
\begin{equation}
  \hat {Z} \hat {X} = \omega \hat {X} \hat {Z}.
\end{equation}

Consistent with~\cite{Almeida98,Rivas99,Kocia16}, we define the symplectic matrix
\begin{equation}
\bs{\mathcal J} =  \left( \begin{array}{cc} 0 & -\mathbb{I}_{n}\\ \mathbb{I}_{n} & 0 \end{array}\right),
\end{equation}
for \(\mathbb{I}_n\) the \(n\)-dimensional identity.

The Weyl symbol (Wigner function) of an operator \(\hat \rho\) can be written
\begin{eqnarray}
  && {\rho}(\bs x_p, \bs x_q)\\
  &=& \Tr \left( d^{-n} \sum_{\bs \lambda \in (\mathbb Z/ d \mathbb Z)^n} e^{\frac{i}{\hbar} {\bs \lambda}^T \bs{\mathcal J} {\bs x} } e^{-\frac{i}{2\hbar} \bs \lambda_p \cdot \bs \lambda_q} \hat Z^{\bs \lambda_p} \hat X^{\bs \lambda_q} \hat \rho \right), \nonumber
  \label{eq:twogenweylsymbol}
\end{eqnarray}
where \(\bs x \equiv (\bs x_p, \bs x_q)\) is a \(2n\)-vector corresponding to a conjugate pair of \(n\)-dimensional \emph{center} momenta and positions while \(\bs \lambda \equiv (\bs \lambda_p, \bs \lambda_q)\) corresponds to its dual \(chord\) momenta and positions.

The Weyl symbol of a unitary gate is in \(\mathbb C\) and any function in \(\mathbb C\) can always be written in the form
\begin{equation}
  \label{eq:weylsymbolofgenunitary}
  U(\bs x) = \exp \left[\frac{2 \pi i}{p^h} S(\bs x) \right],
\end{equation}
where \(S(\bs x) \in \mathbb C[\bs x]\) and \(\sum_{\bs x} U(\bs x) = 1\). Given \(\bs x \in (\mathbb Z / d \mathbb Z)^n\) we can make this more precise by restricting \(S(\bs x)\) to be a polynomial in the ring \(\mathbb C \left[(\mathbb Z / d \mathbb Z)^n\right]\) where \(\deg(S) < d\) (without loss of generality). We call \(S(\bs x)\) the center-generating action or just simply the action due to its semiclassical role~\cite{Almeida98}.

\subsection{Weyl Symbols of Quantum Circuits}

We are interested in the Weyl symbol of products of operators because quantum circuits evolve by products of operators corresponding to elementary gates. The Weyl symbol of a product of an even \(2N\) number of operators \(\hat A_{2N}, \ldots, \hat A_1\) has a phase determined by this symplectic area for the odd number \(2N+1\) of centers:
\begin{eqnarray}
  \label{eq:Weylsymbolofevennumofops}
  &&(A_{2N}\cdots A_1) (\bs x) =\\
  && \frac{1}{d^{2nN}} \sum_{\bs x_1, \ldots, \bs x_{2N}} A_{2N}(\bs x_{2N}) \cdots A_1(\bs x_1) \nonumber \\
                && \times\exp\left[ \frac{i}{\hbar} \Delta_{2N+1}(\bs x, \bs x_1, \ldots, \bs x_{2N}) \right]. \nonumber
\end{eqnarray}
The generalized symplectic matrix for \(N\) degrees of freedom can be defined through the introduction of the matrix \(\mathcal J_N\) and \(H_N\)~\cite{Almeida98}:
\begin{equation}
  \mathcal J_N = \left( \begin{array}{c|c|c|c}\mathcal J & \bs 0& \bs 0& \ldots\\\hline \bs 0& \mathcal J& \bs 0& \\\hline \bs 0& \bs 0& \ddots& \\\hline \vdots & & &\ddots \end{array} \right),
\end{equation}
\begin{equation}
  H_N = \left( \begin{array}{c|c|c|c} \bs 0& -\bs 1& -\bs 1& \ldots\\\hline \bs 1& \bs 0& -\bs{1} & \\\hline \bs 1 & \bs 1& \bs 0 & \\\hline \vdots& & & \ddots \end{array}\right).
\end{equation}
%\pagebreak
The product of \(\mathcal J_N\) and the inverse of \(H_N\) is
\begin{equation}
  \mathcal J_N \mathcal H_N^{-1} = \mathcal H_N^{-1} \mathcal J_N = \left( \begin{array}{c|c|c|c}\bs 0 & \mathcal J & -\mathcal J & \ddots\\\hline -\mathcal J & \bs 0 & \mathcal J & \ddots\\\hline \mathcal J & -\mathcal J & \bs 0 & \ddots\\\hline \ddots & \ddots & \ddots & \ddots \end{array} \right).
\end{equation}

The symplectic area for an odd number of centers is defined to be
\begin{eqnarray}
  \label{eq:symplecticareaodd}
  &&\Delta_{2N+1}(\bs x, \bs x_1, \ldots, \bs x_{2N}) \\
  &=& (\bs x_1 - \bs x, \ldots, \bs x_{2N} - \bs x) \mathcal H^{-1}_{2N} \mathcal J_{2N} (\bs x_1 - \bs x, \ldots, \bs x_{2N} - \bs x). \nonumber
\end{eqnarray}

Notice here the convention of associating a factor of \(d^{-n}\) to every sum over \(\bs x_i \equiv (\bs x_{p_i}, \bs x_{q_i})\). This is a convention that we will continue to use throughout this paper. For instance, we expect
\begin{equation}
  d^{-n} \sum_{\bs x} A(\bs x) = 1,
\end{equation}
for any quantum gate \(\hat A = \hat U\) or density matrix state \(\hat A = \hat \rho\).

For example, the Wigner function of two operators \(\hat A_2 \hat A_1 \) can be written:
\begin{eqnarray}
  \label{eq:Weylsymboloftwoops}
  &&(A_{2} A_1) (\bs x)\\
  &=& \frac{1}{d^{2n}} \sum_{\bs x_1, \bs x_{2}} A_{2}(\bs x_{2}) A_1(\bs x_1) \exp\left[ \frac{i}{\hbar} \Delta_{3}(\bs x, \bs x_1, \bs x_{2}) \right], \nonumber
\end{eqnarray}
where
\begin{eqnarray}
  \Delta_3(\bs x, \bs x_1, \bs x_2) = 2 \bs x^T \bsmc J (\bs x_1 - \bs x_2) + 2 \bs x_1^T \bsmc J \bs x_2,
\end{eqnarray}
is the symplectic area associated with three centers \(\bs x_1\), \(\bs x_2\) and \(\bs x_3\).

Proceeding with a propagator-like treatment, here we will be interested in propagating from some state \(\hat \rho\) to \(\hat \rho'\) and so will be interested in extending these definitions to capture four operators---\(\hat \rho' \hat U \hat \rho \hat U^\dagger\). By Eq.~\ref{eq:Weylsymbolofevennumofops}, this requires the symplectic area for five degrees of freedom:
\pagebreak
\begin{widetext}
  \begin{eqnarray}
    \label{eq:Delta5}
  \Delta_5(\bs x, \bs x_1, \bs x_2, \bs x_3, \bs x_4) &=& (\bs x_1 - \bs x, \bs x_2 - \bs x, \bs x_3 - \bs x, \bs x_4 - \bs x)^T \mathcal H^{-1}_4 \mathcal J_4 (\bs x_1 - \bs x, \bs x_2 - \bs x, \bs x_3 - \bs x, \bs x_4 - \bs x)\\
                                  &=& 2 \bs x_2^T \bsmc J (\bs x_3 - \bs x_1) - 2 \bs x_4^T \bsmc J (\bs x_3 + \bs x_1) + 2 \bs x_4^T \bsmc J \bs x_2 + 2 \bs x^T \bsmc J (\bs x_3 + \bs x_1 - \bs x_4 - \bs x_2) + 2 \bs x_3^T \bsmc J \bs x_1. \nonumber
  \end{eqnarray}
\end{widetext}

More generally, we will be interested in an even numbers of operators: \(\hat \rho' \hat U_1 \cdots \hat U_m \rho \hat U_m^\dagger \cdots \hat U_1\). To accomplish this, we can extend the definition for the Weyl symbol of the product of \(2N\) operators given in Eq.~\ref{eq:Weylsymbolofevennumofops} to an odd number \((2N-1)\) of products by considering the Weyl symbols of \(2N\) products of operators and setting \(\hat A_{2N} = \hat I\), which has the corresponding Weyl symbol \(I(\bs x) = 1\). This leaves the \(\bs x_{2N}\) variable alone in the argument of the exponential, free to be summed over to produce a Kronecker delta function. The remaining exponentiated terms are defined to be \(\Delta_{2N}(\bs x, \bs x_1, \ldots, \bs x_{2N-1})\).

Finally, consider the Weyl symbol of the product of the four operators \(\hat I \hat A_3 \hat A_2 \hat A_1\):
\begin{eqnarray}
  \label{eq:Weylsymbolofthreeops}
  &&A_3 A_2 A_1 (\bs x) =\nonumber\\
  && \frac{1}{d^{2n} }\sum_{\bs x_3, \bs x_2, \bs x_1 \in (\mathbb Z/d \mathbb Z)^{2n}} A_3(\bs x_3) A_2(\bs x_2) A_1(\bs x_1) \\
  && \times \delta(\bs x_3 - \bs x_2 + \bs x_1 - \bs x) \exp\left[ \frac{i}{\hbar} \Delta_4(\bs x, \bs x_1, \bs x_2, \bs x_3) \right], \nonumber
\end{eqnarray}
where
\begin{eqnarray}
  \Delta_4(\bs x, \bs x_1, \bs x_2, \bs x_3) &=& 2 \bs x_2^T \bsmc J \bs x_3 - 2 \bs x_1^T \bsmc J \bs x_3\\
  &&+ 2 \bs x_1^T \bsmc J \bs x_2 + 2 \bs x^T \bsmc J (\bs x_1 + \bs x_3 - \bs x_2). \nonumber
\end{eqnarray}
Notice that the prefactor fell by \(d^{2n}\) compared to Eq.~\ref{eq:Weylsymbolofevennumofops} since that is the dimension of the vector inside the Kronecker delta function. This agrees with our established standard of a factor of \(d^{-n}\) for every sum over \(\bs x_i\) since after the Kronecker delta function is summed over, the resulting expression involves only a double sum over phase space and so we expect two factors of \(d^{-n}\).

With the WWM formalism for the Weyl symbols of products of operators thus established, we move on to consider quantum channels, which are double-ended propagators.

In general, past treatment of the quantum propagation of a state \(\hat \rho\) by a unitary \(\hat U\) in the WWM formalism have been content with Eq.~\ref{eq:Weylsymbolofthreeops}, setting \(\hat A_2 = \hat \rho'\), and \(\hat A_3 = \hat A_1^\dagger = \hat U\)~\cite{Almeida98,Rivas99}. However, such a treatment intimately ties the propagator to the initial state, \(\hat \rho'\), and to act on a final state necessarily involves using Eq.~\ref{eq:Weylsymboloftwoops} setting \(\hat A_2 = \hat \rho\), the final state, and \(\hat A_1 = \hat I \hat U \hat \rho' \hat U^\dagger\). This leaves an intermediary phases \(\Delta_3\) to sum over and as a result does not produce a self-contained double-ended propagator (a quantum channel) that acts on states by summation without any additional functions or phases, as in Eq.~\ref{eq:vVMG}. A double-ended propagator form is often both more familiar and more useful for generality.

Using the additional general formulae defined in this Section, notably Eq.~\ref{eq:Weylsymbolofevennumofops}, we can develop such a double-ended propagator or quantum channel.

\subsection{Weyl Symbols for Quantum Channels}

We define \(U U^*(\bs x, \bs x')\) to be the Weyl symbol of a quantum channel that can take any initial state \(\rho'\) to a final states \(\rho\) by summing over their phase space variables:
\begin{equation}
  \label{eq:quantumchannel}
  (\rho U \rho' U^\dagger)(\bs x) \equiv \frac{1}{d^{2n}} \sum_{\bs x, \bs x' \in (\mathbb Z/ d \mathbb Z)^{2n}} \rho(\bs x) UU^*(\bs x, \bs x') \rho'(\bs x').
\end{equation}

To obtain \(U U^*(\bs x, \bs x')\), we use Eq.~\ref{eq:Weylsymbolofevennumofops} for \(2N=4\). We set \(A_4(\bs x_4) = 1\) and sum \(\bs x_4\) away, set \(A_2(\bs x_2) = \rho'\) but then discard it along with its sum over \(\bs x_2\), and set \(A_3(\bs x) = A_1^*(\bs x) = U(\bs x)\). Relabelled the variables of summation, this produces:
\begin{eqnarray}
  U U^*(\bs x, \bs x') &\equiv& \frac{1}{d^{2n}} \sum_{\bs x_1, \bs x_2, \bs x_3 \in (\mathbb Z/p \mathbb Z)^{2n}} U(\bs x_1) U^*(\bs x_2) \nonumber\\
  \label{eq:Weylprop}
                       && \times\exp\left[ \frac{2 \pi i}{d} \Delta_5(\bs x_3, \bs x, \bs x_1, \bs x', \bs x_2)\right],
\end{eqnarray}
where we also used the identity \(\Delta_5(\bs x, \bs x_1, \bs x', \bs x_2, \bs x_3) = \Delta_5(\bs x_3, \bs x, \bs x_1, \bs x', \bs x_2)\). This is equivalent to using Eq.~\ref{eq:Weylsymbolofthreeops} for \(A_2(\bs x_2) = \rho'(\bs x_2)\) and discarding the sum over \(\bs x_2\).

Notice in Eq.~\ref{eq:Weylprop} the inclusion of the full phase \(\Delta_5\) in the expression. As a result, acting on \(\rho(\bs x)\) or \(\rho'(\bs x')\) by \(U U^*(\bs x, \bs x')\) as in Eq.~\ref{eq:Weylprop} simply involves summing over \(\bs x\) or \(\bs x'\) respectively; there is no intermediate phase \(\Delta_3\) as in Eq.~\ref{eq:Weylsymboloftwoops} that must also be summed over. 

Our choice of normalization in Eq.~\ref{eq:Weylprop} allows us to continue the standard of including a factor of \(d^{-n}\) to every sum over a pair of conjugate phase space degrees of freedom in Eq.~\ref{eq:quantumchannel}. This means that \(d^{-2n} \sum_{\bs x, \bs x' \in (\mathbb Z/d \mathbb Z)^{2n}} UU^*(\bs x, \bs x') = d^{2n}\) while \(d^{-n} \sum_{\bs x \in (\mathbb Z/d \mathbb Z)^{2n}} UU^\dagger(\bs x) = d^n = \Tr( \hat U \hat U^\dagger)\).

The notation \(U U^*(\bs x, \bs x')\) for the Weyl symbol of this propagator is perhaps clunky, since it can be mistaken for the Weyl symbol of \(\hat U \hat U^\dagger\), \(U U^\dagger(\bs x)\). But as Eq.~\ref{eq:Weylprop} denotes, it is the Weyl symbol of a double-ended propagator and so has two arguments, \(\bs x\) \emph{and} \(\bs x'\)\footnote{To avoid any confusion, we point out that \(U U^*(\bs x, \bs x') \ne U U^\dagger(\bs x) = 1\); \(U U^\dagger(\bs x)\) is the Weyl symbol of \(\hat U \hat U^\dagger = \hat I\) and so is a function of only one variable, \(\bs x\); % (although it turns out to be independent of it); 
\(U U^*(\bs x, \bs x')\) can most appropriately be associated with the Weyl symbol of the superoperator \(\odot \hat U \bullet \hat U^\dagger\), where \(\odot\) and \(\bullet\) denote the operators the superoperator acts on, and so is a function of two variables, (cont'd on pg. \(7\)) (cont'd from pg. \(5\)) denoted \(\bs x\) and \(\bs x'\).\newline
\indent As added incentive, \(U U^*(\bs x, \bs x')\) is naively a simpler function to deal with than the Weyl symbol of a unitary operator, \(U(\bs x)\). \(U U^*(\bs x, \bs x')\) is real-valued just like the Weyl symbols (\(\rho(\bs x)\)) of density functions (Wigner functions), while \(U(\bs x)\) is generally complex-valued. Also, \(U U^*(\bs x, \bs x')\) resembles traditional propagators more closely in that it can be said to take states from \(\bs x'\) to \(\bs x\), whereas \(U(\bs x)\) requires pairing with \(U^*(\bs x)\) and summation over intermediate values with the appropriate phase as in Eq.~\ref{eq:Weylsymbolofthreeops}, to act on a state.}

For Clifford gates \(\hat U\), the associated Weyl symbol \(U(\bs x)\) is a non-negative map; \(U(\bs x)\) takes non-negative states to non-negative states. This is clearer for \(U U^*(\bs x, \bs x')\), which becomes a non-negative real function: \(U U^*(\bs x, \bs x') \ge 0\). This parallels the non-negativity of Wigner functions \(\rho(\bs x)\) of stabilizer states \(\hat \rho\). The extension of this set by \(\frac{\pi}{8}\) gates or magic states to allow for universal quantum computation introduces contextuality into \(U U^*(\bs x, \bs x')\) and \(\rho(\bs x)\), respectively, which means that they can now also have (real) negative values.

In the next section we will explore how the non-negative real function \(U U^*(\bs x, \bs x')\) produced by a Clifford gate can be described by a single Gauss sum.

\section{Gauss Sums}
\label{sec:Gausssums}

We now restrict to \(d=p^h\) for \(p\) odd prime and \(h \in \mathbb Z\).

Consider, a simplified definition of the Gauss sum from~\cite{Dabrowski97,Fisher02}:
\begin{eqnarray}
  G_h(\bs A, \bs v) = p^{nh/2} \sum_{\bs x \in (\mathbb Z/p^h \mathbb Z)^n} &&\exp \left[ \frac{2\pi i}{p^h} \left( \bs x^T \frac{\bs A}{2} \bs x \right) \right] \nonumber\\
  \label{eq:Gausssum}
                    && \times \exp \left[ \frac{2 \pi i}{p^h} \bs v^T \bs x \right],
\end{eqnarray}
for \(\bs A \in \mathbb Z^{n \times n}\), \(\bs v \in \mathbb Z^n\), and where for \(h\le 0\), \(G_h(\bs A, \bs v) = p^{nh/2}\).

The result of evaluating this sum can be simplified and summarized in the following list~\cite{Dabrowski97}:
\begin{proposition}{Gauss Sums}
  \label{prop:Gausssums}
\begin{enumerate}
\item[(a)] If \(\det \bs A \ne 0\) and \(h\) is large enough that \(\bs A' = p^h \bs A^{-1}\) has entries in \(\mathbb Z/p \mathbb Z\). 
Then \(G_h(\bs A, \bs 0)\) is \(p^{\frac{1}{2}v_p(\det \bs A)}\) times a fourth root of unity, where \(v_p(\alpha)\) is the highest exponent such that \(p^{v_p(\alpha)}\) divides \(\alpha\).
\item[(b)] If \(\exists \, \bs u \in \mathbb Z / p^n \mathbb Z\) such that \(\bs v = \bs A \bs u\) then \(G_h(\bs A, \bs v) = G_h(\bs A) \exp \left( -\frac{\pi i}{p^h} \bs u^t \bs A \bs u \right)\). Otherwise, \(G_h(\bs A, \bs v) = 0\).
\item[(c)] Either \(G_1(\bs A, \bs v) = 0\) or it is \(p^{(n-r)/2}\) times a root of unity as in (a), where \(r\) is the rank of \(\bs A\). 
\end{enumerate}
\end{proposition}

There is a more broad application of this result to \(p\)-adic number theory, a topic that is briefly introduced in the Appendix~\ref{app:padicnumbers}, but is beyond the scope of this work.

A simple case of this result produces one of the most often used equations in this paper:
\begin{corollary}[Kronecker Delta Function]
\label{corr:Gausssums}
  \begin{equation}
    p^{-n} \sum_{\substack{\bs x \in\\ (\mathbb Z/ p \mathbb Z)^n}} \exp \left( \frac{2 \pi i}{p} \bs v \cdot \bs x \right) = \prod_{i=1}^n \delta (v_i \Mod p) \equiv \delta(\bs v),
  \end{equation}
\end{corollary}
\begin{proof}
  \(\bs A = \bs 0\) and so by Proposition~\ref{prop:Gausssums}(b), the sum is equal to \(0\) if \(\bs v \ne \bs 0\) since there does not exist any \(\bs u \in \mathbb Z/ p^n \mathbb Z\) such that \(\bs A \bs u \ne 0\). On the other hand, if \(\bs v = \bs 0\), then Proposition~\ref{prop:Gausssums}(b) states that the sum is equal to \(p^{-3n/2} G_1(0) \exp \left( -\frac{\pi i}{p} \bs u^t \bs A \bs u \right) = 1 \times \exp \left( -\frac{\pi i}{p} \bs v^t \bs v \right) = 1\). This is the definition of the Kronecker delta function.
\qed
\end{proof}

We now give an example of the application of this Proposition to Clifford gates.

\subsection{Gauss Sums on Clifford Gates}
\label{subsec:GausssumsonCliffordgates}

As we saw in Eq.~\ref{eq:weylsymbolofgenunitary}, the Weyl symbol of a unitary gate such as a Clifford gate is \(U(\bs x) = \exp \left[\frac{2 \pi i}{p^h} S(\bs x) \right]\). Clifford gates have actions that fall into a subring of \(\mathbb C \left[(\mathbb Z / d \mathbb Z)^n\right]\): \(S(\bs x) \in \mathbb Z / d \mathbb Z \left[(\mathbb Z / d \mathbb Z)^n\right]\) where \(\deg (S) \le 2\)~\cite{Kocia16}. We can write the action for a Clifford gate as:
\begin{equation}
  S(\bs x) = \bs x^T \bs B \bs x + \bs \alpha^T \bsmc J^T \bs x,
\end{equation}
where \(\bs B\) is a \(2n \times 2n\) symmetric matrix with elements in \(\mathbb Z/d\mathbb Z\).

This simpler form makes Clifford gates act on states in a correspondingly simpler manner. Namely, Clifford gates transform Wigner functions by a linear symplectic transformation~\cite{Kocia16}. One way to see this is by deriving functions that act as Hamiltonians from the above action, which can be transformed into the computational basis by Legendre transform~\cite{Kocia16}. These Hamiltonians are similarly harmonic and so their effect on states can be described that the action of harmonic Hamiltonians acting on continuous degrees of freedom can be: by rigid symplectic transformations. In the discrete case, this translated to taking Wigner arguments (\(\bs x\)) to themselves while preserving area (a symplectic linear transformation). We proceed to see how this is manifested from the perspective of the double-ended propagator \(U U^*(\bs x, \bs x')\).

Let us evaluate \(U U^*(\bs x, \bs x')\) from this perspective:
\begin{eqnarray}
  U U^*(\bs x, \bs x') &=& \frac{1}{p^{2n}} \sum_{\bs x_1, \bs x_2, \bs x_3 \in (\mathbb Z/p \mathbb Z)^{2n}} U(\bs x_1) U^*(\bs x_2) \nonumber\\
                       && \times\exp\left[ \frac{2 \pi i}{p} \Delta_5(\bs x_3, \bs x, \bs x_1, \bs x', \bs x_2)\right],
\end{eqnarray}
where \(U(\bs x) = \exp\left[ \frac{2 \pi i}{p} \left(\bs x^T \bs B \bs x + \bs \alpha^T \bsmc J^T \bs x \right) \right]\).

\begin{widetext}
Applying Proposition~\ref{prop:Gausssums} produces
\begin{equation}
  U U^*(\bs x, \bs x') = G_1(\bs A, \bs 0) \sum_{\bs x_3 \in (\mathbb Z/p \mathbb Z)^{2n}} \exp \left[ -\frac{\pi i}{p} \left( \bs x_3^T \bs A' \bs x_3 + 2 \bs v'^T \bs x_3 + 2 c' \right) \right].
\end{equation}
where \(\bs A' = 0\), \(2 \bs v'^T = (2 \bs x' + (\bs \alpha \Mod p))^T (4 \bsmc J) (\bsmc J \bs B - \bs I)^{-1} + 4 (\bs x + \bs x')^T \bsmc J\), and \(2 c' = 4 \bs x'^T [4 \bs \beta \bsmc J - 4 \bs \beta \bsmc J \bs B^{-1} \bsmc J - \bsmc J] \bs x + 8 \bs \alpha^T \bs \beta \bsmc J [1 - \bs B^{-1} \bsmc J] \bs x\) for \(\bs \beta = -\frac{1}{2} \bsmc J^T (\bs B - \bsmc J \bs B^{-1} \bsmc J)^{-1} \) (see Appendix~\ref{app:Cliffgate}).
\end{widetext}

We can define a symplectic matrix \(\bsmc M\) by relating it to the symmetric matrix \(\bs B\) through the Cayley parametrization:
\begin{equation}
  \bsmc M \equiv (1 + \bsmc J \bs B)^{-1} (1 - \bsmc J \bs B) = (1 - \bsmc J \bs B) (1 + \bsmc J \bs B)^{-1}.
\end{equation}
By Corollary~\ref{corr:Gausssums}, since \(\bs A' = \bs 0\), the sum over \(\bs x_3\) is non-zero if and only if \(\bs v' = 0\) (see Appendix~\ref{app:Cliffgate}):
\begin{eqnarray}
  \bs x &=& \bsmc M \left(\bs x' + \frac{\bs \alpha}{2}\right) + \frac{\bs \alpha}{2}.
\end{eqnarray}

Or, in other words, we are presented with the usual plane wave sum identity for a Kronecker delta function.

For values of \(\bs x'\) and \(\bs x\) that satisfy \(\bs x = \bsmc M (\bs x' + \frac{\bs \alpha}{2}) + \frac{\bs \alpha}{2}\), \(U U^*(\bs x, \bs x') = d^{2n}\). There is only one solution \(\bs x\) given \(\bs x'\), and so \(d^{-2n} \sum_{\bs x, \bs x' \in (\mathbb Z/d \mathbb Z)^{2n}} UU^*(\bs x, \bs x') = d^{2n}\). This result affirms that the Clifford gate symplectically transforms Wigner functions by point-to-point symplectic (area-preserving) permutation as discussed in the beginning of this Section.

Note that this approach is no longer suitable if \(S(\bs x)\) has powers higher than quadratic, or, if even though it is quadratic, \(\bs B\) and \(\bs \alpha\) do not have elements in \(\mathbb Z\). This is because the Gauss sum results from Proposition~\ref{prop:Gausssums} cannot apply in these cases. To handle a more general form for \(S(\bs x)\), the the Gauss sum results of this section must be extended. We explore such a generalization in Section~\ref{sec:statphase} where we develop the discrete and periodic stationary phase method.

\section{Stationary Phase}
\label{sec:statphase}

We will now proceed to momentarily consider the more general case of a sum over an exponentiated polynomial \(S(\bs x)\) that lies in \(\mathbb Q_p[x_1, \ldots, x_n]\), the field of p-adic numbers (see Appendix~\ref{app:padicnumbers}), instead of in \(\mathbb Z/p \mathbb Z[x_1, \ldots, x_n]\), but such that \(\partder{S(\bs x)}{\bs x}\) still has coefficients in \(\mathbb Z_p\). We do this so as to make use of a result from this more general treatment since, in this case, we can no longer derive a simple Gauss sum result. However, if we split the sum up, we can obtain a closed form in terms of several Gauss sums or sums of higher order with arguments in \(\mathbb Z/p \mathbb Z[x_1, \ldots, x_n]\) again. 

To accomplish this we express the sum, \(\mathcal I\), in terms of local terms \(\mathcal I_{\bar{\bs x}}\):
\begin{eqnarray}
  \label{eq:expsum}
  \mathcal I &=& \sum_{\bs x \in (\mathbb Z/p^m \mathbb Z)^n} \exp \left[ \frac{2 \pi i}{p^m} S(\bs x) \right]\\
             &=& \sum_{\bar{\bs x} \in (\mathbb Z/p^j \mathbb Z)^n} \left\{ \sum_{\substack{{\bs x} \in (\mathbb Z/p^m \mathbb Z)^n\\\bs x \Mod p^j = \bar{\bs x}}} \exp \left[ \frac{2 \pi i}{p^m} S(\bs x) \right] \right\} \nonumber\\
             &\equiv& \sum_{\bar{\bs x} \in (\mathbb Z/p^j \mathbb Z)^n} \mathcal I_{\bar{\bs x}}, \nonumber
\end{eqnarray}
where \(j < m\).

So far this is simply a reorganization of the order of terms in the summation. However, since these local terms are over coarser periodic domains, their polynomial power can be reduced in much the same way that polynomial powers can be lowered in proximal Taylor series expansions over a continuous domain. In this case, since we are dealing with a periodic discrete domain, such a simplification relies on the ``polynomial version'' of Taylor's theorem~\cite{Bourbaki90},
\begin{equation}
  S(a + p^j x) = \sum_{\alpha} \frac{1}{\alpha !} \frac{\partial^\alpha}{\partial x^\alpha} \left[ S(x) \right]_{x = a} \cdot p^{j|\bs \alpha|} x^\alpha,
\end{equation}
where the sum is over the components of \(0 \le \alpha_i \in \mathbb Z^n\) such that \(|\bs \alpha| \equiv \prod_i (\alpha_i !)\).
Using the polynomial version of Taylor's theorem, Fisher proved the following Lemma~\cite{Fisher02}:
\begin{lemma}[Fisher1]
  \label{lem:fisher1}
  Given positive integers \(j\) and \(k\), a polynomial \(S(x) \in \mathbb Q_p[x_1, \ldots x_n]\), and a point \(a \in \mathbb Z^n_p\), let \(S_a^{(<k)}(x)\) denote the Taylor polynomial of degree \(k-1\) for \(S(a+x)\). Assume that \(p^t S(x)\) and the partial derivatives have coefficients in \(\mathbb Z_p\), for some integer \(t\), and let
  \begin{equation}
    \label{eq:epsilonk}
    \epsilon = \epsilon_k = \min\{t, \lfloor \log k/\log p \rfloor \}.
  \end{equation}
  Then \(S(a + p^j x) = S_a^{(<k)}(p^j x) \Mod p^{jk-\epsilon}\) as polynomials in \(x\).
\end{lemma}

We express the quadratic Taylor polynomial of \(S(x)\) at \(a\) as \(S_a^{<3}(x) = S(a) + \nabla S(a) \cdot x + \frac{1}{2} x \cdot H_a\cdot x\), where \(H_a= \left.\frac{\partial^2 S}{\partial x_\mu \partial x_\nu}\right|_{x = a}\) is the Hessian matrix of \(f\) at \(a\).

We proceed as follows:
\begin{theorem}[Fisher2]
  \label{th:statphase}
  Let \(S(x) \in \mathbb Q_p[x_1, \ldots, x_n]\) such that \(\nabla S(x) = (\partder{S}{x_1}, \ldots, \partder{S}{x_n})\) has coefficients in \(\mathbb Z_p\). Define \(\epsilon_k\) as in Eq.~\ref{eq:epsilonk}. Given integers \(m \ge j \ge 1\) and a point \(a \in \mathbb Z^n_p\) with reduction \(\bar a \in (\mathbb Z/p^j \mathbb Z)^n\), define \(\mathcal I_{\bar a} = \mathcal I_{\bar a} (f; \mathbb Z/p^m \mathbb Z)\) as in Eq.~\ref{eq:expsum}. Then:
  \begin{enumerate}
  \item[(a)] The local sum \(\mathcal I_{\bar a}\) vanishes unless \(\nabla S(a) \equiv 0 \Mod p^{\min (j, m-j)}\).
  \item[(b)] If \(m \le 3j - \epsilon_3\) and \(\nabla S(a) \equiv 0 \Mod p^{\min(j, m-j)}\), then
    \begin{equation}
      \mathcal I_{\bar a} = p^{nm/2}e^{2 \pi i S(a)/p^m} G_{m-2j}(H_a,p^{-j}\nabla S(a)).
    \end{equation}
  \item[(c)] Assume that \(\det (H_a) \ne 0\), and choose an integer \(h\) such that \(p^h H_a^{-1}\) has coefficients in \(\mathbb Z_p\). Assume that \(j > h\) and \(2j + h \le m \le 3j - \epsilon_3\). Then the following are equivalent:
    \begin{enumerate}
    \item[(i)] \(\mathcal I_{\bar a} \ne 0\),
    \item[(ii)] \(\nabla S(a) \in p^j H_a \cdot \mathbb Z^n_p\), i.e., \(p^j H_a | \nabla S(a)\),
    \item[(iii)] \(\exists \alpha \in \mathbb Z^n_p: \nabla S(a) = 0 \in \mathbb Z^n_p\) and \(\alpha \equiv a \Mod p^j\), and \(\mathcal I_{\bar a} = p^{n m/2} e^{2 \pi i S(\alpha)/p^m} G_m(H_\alpha,0)\).
    \end{enumerate}
  \end{enumerate}

\end{theorem}

As an example consider \(S(x) = x^3 + 2 x^2\) over \(\mathbb Z/ 3^2 \mathbb Z\):
\begin{equation}
  \mathcal I = \sum_{x \in \mathbb Z/ 3^2 \mathbb Z} \exp\left[\frac{2 \pi i}{3^2} (x^3+2x^2)\right].
\end{equation}
We reexpress this as a sum of local terms \(\mathcal I_{\bar x}\):
\begin{equation}
  \mathcal I = \sum_{\bar x \in \mathbb Z/ 3 \mathbb Z} \mathcal I_{\bar x},
\end{equation}
where
\begin{equation}
 \mathcal I_{\bar x} = \left\{ \sum_{\substack{\bs x \in \mathbb Z/3^2 \mathbb Z\\ \bs x \Mod p = \bar x}} \exp\left[ \frac{2 \pi i}{3^2} (x^3+2x^2) \right] \right\}.
\end{equation}
\vspace{5pt}
Since \(S(x) \in \mathbb Z_p\) it follows that \(t=0\) and so \(\epsilon=0\).

\(0 \Mod 3 = \nabla S(x) = 3 x^2 + 4 x\) has solution \(x = 0\). Therefore, by Theorem~\ref{th:statphase}(a),
\begin{equation}
  \mathcal I = \mathcal I_{\bar 0}.
\end{equation}
Furthermore, by Theorem~\ref{th:statphase}(b), since \(2 = m \le 3j - \epsilon_3 = 3\), it follows that
\begin{equation}
  \mathcal I = \mathcal I_{\bar 0} = 3 e^{2 \pi i S(0)/3^2} = 3.
\end{equation}

Therefore, due to Theorem~\ref{th:statphase}, we are able to simplify the sum, \(\mathcal I\), to just one term, \(\mathcal I_{\bar 0}\).

\subsection{Diagonal Unitary Example}

An important subset of non-Clifford gates include diagonal gates with rational eigenvalues. These can always be written as \(\hat U = \exp \left[ \frac{2 \pi i}{p^h} S(\hat q) \right]\) where \(S(\hat q) \in \mathbb Q[\hat q]\). Hence, their Weyl symbols are \(U(\bs x_q) = \exp \left[ \frac{2 \pi i}{p^h} S(\bs x_q) \right]\).

As an example, we consider an action that encompasses the generalized \(\frac{\pi}{8}\)-gates for qutrits~\cite{Campbell12,Howard12}:
\begin{equation}
  S_9(\bs x) = C x_q^3 + B x_q^2 + \bs \alpha \bsmc J^T \bs x,
\end{equation}
where \(C, B \in \mathbb Z\) and \(\bs \alpha \in \mathbb Z^n\). We consider \(S_9\) as a polynomial over \(\mathbb Z/3^2\mathbb Z\) so that its Weyl symbol is \(U_9(\bs x) \equiv \exp\left(\frac{2 \pi i}{3^2}S_9(\bs x)\right)\). We note that though the coefficients of this polynomial are still in \(\mathbb Z\) as the actions were in Section~\ref{sec:Gausssums}, it is cubic and thus must be treated by the method of stationary phase instead of as a single quadratic Gauss sum.

It is easy enough to verify that this corresponds to a unitary operator (see Appendix~\ref{app:firstnontrivialex}):
\begin{widetext}
\begin{equation}
  \left(U_9 U_9^* \right) (\bs x) = \left(\frac{1}{3^2}\right)^2\sum_{\bs x', \bs x'' \in (\mathbb Z/3^2\mathbb Z)^2} U_9(\bs x'') U_9^*(\bs x') \exp\left(\frac{2 \pi i}{3^2} \Delta_3(\bs x, \bs x', \bs x'') \right) = 1.
\end{equation}
So we now consider the double-ended propagator corresponding to \(S_9\):
\begin{eqnarray}
  U_9 U_9^*(\bs x, \bs x') &=& \frac{1}{9^{2}} \sum_{\bs x_1, \bs x_2, \bs x_3 \in (\mathbb Z/3^2 \mathbb Z)^{2N}} U(\bs x_1) U^*(\bs x_2) \exp\left[ \frac{2 \pi i}{p} \Delta_5(\bs x_3, \bs x, \bs x_1, \bs x', \bs x_2)\right] \\
                           &=& 9 \sum_{\{\tilde x_{2q}\} \Mod 3} 3 \exp\left[\frac{2 \pi i}{3^2} S(\tilde x_{2q}) \right], \nonumber
\end{eqnarray}
where \(S(\tilde x_{2q})\) is a cubic polynomial shown in Appendix~\ref{app:firstnontrivialex}. 
\end{widetext}

For some values of \(A,\, B\), and \(\bs \alpha\), some phase space points \((\bs x, \bs x')\) produce three Gauss sums while others produce none. This can be found by seeing which values of \((\bs x, \bs x')\) cause \(\partder{S(x_{2q})}{x_{2q}}\) to equal zero. Examples of this can be found in Appendix~\ref{app:firstnontrivialex}. This illustrates the usefulness of the periodized stationary phase approximation in not only formalizing the number of Gauss sums that are necessary, but also in providing an easy way to find this number (i.e. finding the zeros of \(\partder{S(x_{2q})}{x_{2q}}\)).

For general diagonal gates, the coefficients in their action polynomials with fall in \(\mathbb Q\) instead of \(\mathbb Z\). We will develop a technique in Section~\ref{sec:qutritpi8gate} that will show how to change the domain of summation from \(\mathbb Z / p^h \mathbb Z\) to \(\mathbb Z/ p^{h'} \mathbb Z\), where \(h'>h\), such that the resultant equivalent action has coefficients in \(\mathbb Z\) thereby allowing for its treatment by the method of periodized stationary phase developed here. But first, let us characterize and discuss this periodized stationary phase method.

\section{Uniformization}
\label{sec:uniformization}

The discrete ``periodized'' stationary phase method introduced in the Section~\ref{sec:statphase} has many similarities to the stationary phase approximation in the continuous setting. This is discussed in more detail in Appendix~\ref{sec:periodizedstatphase}. One of the properties of the continuous stationary phase approximation is that the number of Gaussian integrals it produces can be decreased by increasing the expansion in the exponentiated action to a higher order than quadratic. This produces a smaller number of higher order integrals, such as Airy functions (corresponding to a cubic expansion) and the Pearcey integral (corresponding to a quartic expansion). This procedure is called uniformization. Though this produces fewer terms, these integrals are correspondingly more difficult to evaluate compared to Gaussian integrals.

A similar uniformization procedure can be applied with the discrete ``periodized'' stationary phase method. As we shall see here, this reduces the number of quadratic Gauss sums that the method produces into a fewer number of higher-order Gauss sums.

According to Theorem~\ref{th:statphase}, the discrete stationary phase method drops the size of summation by a factor of \(p^{m - \lceil m/3 \rceil}\) and changes the summand to a sum of quadratic Gauss sums with closed-form solutions multiplied by phases. However, the stationary phase method can also decrease the domain by a larger multiple of \(p\), but the resultant sum is over terms that are themselves sums with higher than quadratic order; they are no longer quadratic Gauss sums.

To prove part (b) of Theorem~\ref{th:statphase}, Fisher appealed to Lemma~\ref{lem:fisher1} to show that if \(m \le kj-\epsilon_k = 3 j - \epsilon_3\) and \(\nabla S(a) \equiv 0 ( \Mod p^{\min \{j, m-j\}})\), then for any \(x \in (\mathbb Z/ p^{m-j} \mathbb Z)^n\),
\begin{equation}
  S(a + p^j x) = S(a) + p^j \nabla S(a) \cdot x + \frac{1}{2} p^{2j} H_a(x) \in \mathbb Q_p/p^m \mathbb Z.
\end{equation}
This implies
\begin{eqnarray}
  I_{\bar a} &=& \sum_{x \in (\mathbb Z/ p^{m-j} \mathbb Z)^n} e^{2 \pi i S(a + p^j x)/p^m} \\
  &=& e^{2 \pi i S(a)/p^m} \sum_{x \in (\mathbb Z/ p^{m-j} \mathbb Z)^n} e^{2 \pi i \left[ p^{-j} \nabla S(a)\cdot x + \frac{1}{2} H_a(x) \right]/p^{m-2j}}, \nonumber
\end{eqnarray}
which can be rewritten in terms of the Gauss sum notation as above in Theorem~\ref{th:statphase} (b).

  This trivially generalizes to higher order polynomials:
\begin{theorem}[Generalization of Fisher2 (b)]
  \label{th:statphasegeneral}
  Let \(S(x) \in \mathbb Q_p[x_1, \ldots, x_n]\) such that \(\nabla S(x) = (\partder{S}{x_1}, \ldots, \partder{S}{x_n})\) has coefficients in \(\mathbb Z_p\). Define \(\epsilon_k\) as in Eq.~\ref{eq:epsilonk}. Given integers \(m \ge j \ge 1\) and a point \(a \in \mathbb Z^n_p\) with reduction \(\bar a \in (\mathbb Z/p^j \mathbb Z)^n\), define \(\mathcal I_{\bar a} = \mathcal I_{\bar a} (f; \mathbb Z/p^m \mathbb Z)\) as in Eq.~\ref{eq:expsum}. Then:
  If \(m \le kj - \epsilon_k\) and \(\nabla S(a) \equiv 0 \Mod p^{\min(j, m-j)}\), then
  \begin{equation}
  I_{\bar a} = \sum_{x \in (\mathbb Z/ p^{m-j} \mathbb Z)^n} e^{2 \pi i S(a + p^j x)/p^m} = \sum_x e^{2 \pi i S^{<k}(a)/p^{m-j}}.
  \end{equation}
\end{theorem}  
\begin{proof}
  The same proof as Fisher's can be used here. \qed
\end{proof}

In particular, for \(k=4\),
\begin{equation}
  I_{\bar a} = e^{2 \pi i S(a)/p^m} \sum_x e^{2 \pi i \left[ p^{-2j} \nabla S(a)\cdot x + \frac{1}{2} p^{-j} H_a(x) + \frac{1}{3!} A_a(x)\right]/p^{m-3j}},
\end{equation}
where
\begin{equation}
  A_a = \sum_{i,j,k}\left. \frac{\partial^3 S(x)}{\partial x_i \partial x_j \partial x_k} \right|_{x = a} x_i x_j x_k.
\end{equation}
Completing the cube for \(n=1\) by setting \(x' = \left(\nabla^3 S\right)^{\frac{1}{3}} \left(x + p^{-j} \nabla^2 S/ \nabla^3 S\right)\) this can be rewritten as:
\begin{equation}
  I_{\bar a} = e^{2 \pi i (S(a)+ \Gamma)/p^m} \sum_{x'} e^{2 \pi i \left[ \Xi x' + x'^3 \right]/p^{m-3j}},
\end{equation}
for
\begin{eqnarray}
  \Gamma &=& - p^{-3j} \nabla S(a) (\nabla^2 S)/(\nabla^3 S) + \frac{1}{2} p^{-3j} (\nabla^2 S)^3/(\nabla^3 S) \nonumber\\
  && -\frac{1}{6} p^{-3j} (\nabla^2 S)^3/(\nabla^3 S)^2, \\
  \Xi &=& \left(\nabla^3 S\right)^{-\frac{1}{3}} \left[ p^{-2j} \nabla S(a) - p^{-2j} (\nabla^2 S)^2/(\nabla^3 S) \right. \nonumber\\
  && \left. +\frac{1}{2} p^{-2j} (\nabla^2 S)^2/(\nabla^3 S) \right].
\end{eqnarray}
(Note that \(\left( \nabla^3 S \right)^{\frac{1}{3}}\) is well-defined for prime dimension.)

This is a discrete sum analog to the \emph{Airy function}~\cite{Olver10}.

Similarly, for \(k=5\), setting \[x' = \left(\nabla^4 S\right)^{\frac{1}{4}} \left( x + (p^{-j} \nabla^3 S/ \nabla^4 S) \right)\] for \(n=1\) allows the equation to be rewritten as:
\begin{equation}
  S_{\bar a} = e^{2 \pi i (S(a)+ \Gamma)/p^m} \sum_{x'} e^{2 \pi i \left[ \Xi x' + \Upsilon x'^2 + x'^4 \right]/p^{m-4j}},
\end{equation}
for
\begin{eqnarray}
  \Upsilon &=& \left(\nabla^4 S\right)^{-\frac{1}{2}} \left[ \frac{1}{2} p^{-2j} \nabla^2 S \right. \nonumber\\
           && \left.- \frac{1}{4} \left(p^{-j} \nabla^3 S \right)^2 / \left(\nabla^4 S \right) \right],\\
  \Xi &=& \left(\nabla^4 S\right)^{-\frac{1}{4}} \left[p^{-3j} \nabla S - \frac{1}{6} p^{-2j} (\nabla^3 S)^2 / (\nabla^4 S) \right.\nonumber\\
           && \left. - \left(\nabla^4 S\right)^{\frac{1}{2}} 2 p^{-j} (\nabla^3 S) / (\nabla^4 S) \Upsilon \right], \\
  \Gamma &=& - (\frac{1}{24} \nabla^4 S) (p^{-j} \nabla^3 S/ \nabla^4 S)^4 - (p^{-j} \nabla^3 S/ \nabla^4 S)^2 \Upsilon \nonumber\\
           && - (p^{-j} \nabla^3 S/ \nabla^4 S) \Xi.
\end{eqnarray}
This is a discrete sum analog to the \emph{Pearcey integral}~\cite{Pearcey46}.

Higher order instances can be similarly developed leading to discrete sum analogs to integrals familiar in catastrophe theory~\cite{Gilmore93}.

For \(n\) degrees of freedom with a domain of summation of \(\mathbb Z/ p^m \mathbb Z\) (as might be produced by, for instance, \(n\) \(p^m\)-dimensional qudits), the naive numerical summation involves a sum over \(p^{n m}\) terms. Using the Gauss sum \(k=3\) simplification or appropriate uniformization level for \(k > 3\), this sum can be reduced to a sum over \(p^{n}\) terms involving Gauss, Airy, Pearcey etc. sums, which number \(\prod_j p^{m-j}\) and can be pre-computed and stored for use during the summation. These terms can perhaps also be approximated numerically instead of tabulated, thereby eschewing the exponential cost in storage. This is the current approach with Airy functions and Pearcey integrals in the continuous case in computation, for instance~\cite{Chester57}.

We note that including non-Clifford gates within the WWM formalism was not possible under the previous derivation of the formalism in terms of powers of \(\hbar\)~\cite{Almeida98,Rivas99,Kocia16}. This is because earlier derivations of the discrete WWM formalism begain with the continuous case, where the stationary phase method is able to be applied in its traditional form, and then ``discretized'' and ``periodized'' the final propagator. As a result, the propagator could only be written to order \(\hbar\), which just includes Clifford propagation, and higher order corrections could not be derived because anharmonic trajectories cannot be obtained by ``periodizing'' to a discrete Weyl phase space grid. In this way, this paper finally accomplishes this extension to higher orders of \(\hbar\) by instead using this different ``periodized'' stationary phase method, thereby solving this old problem for the first time. Just as in the old derivation, it finds that Clifford propagation is captured at order \(\hbar^0\), but unlike the old derivation it is able to formally ascribe a power of \(\hbar\) required to include (diagonal and their Clifford transformations) non-Clifford gates in a formal Taylor series, as discussed in Section~\ref{sec:periodizedstatphase}.

\section{Application to universal gatesets}
\label{sec:qutritpi8gate}

The \(T\)-gate is a non-Clifford single qudit gate. It extends the Clifford gateset to allow for universal quantum computation~\cite{Campbell12}. Generalizations of the \(T\)-gate to qudits are frequently called \(\pi/8\) gates.

Here we will show how the discrete ``periodic'' stationary phase method can be applied to the Wigner function of the qutrit \(\pi/8\) gate magic state to obtain a summation over quadratic Gauss sums:
  \begin{eqnarray}
    &&\rho_{\pi/8}(\bs x) =\\
    &&\frac{1}{3^2} \sum_{x_{2q} \in \mathbb Z/ 3^2 \mathbb Z} \exp \left\{\frac{2 \pi i}{3^2} S''_{\pi/8}(x_{2q}, x_q, x_p, x'_p {=} 0) \right\}, \nonumber
  \end{eqnarray}
  where \(S''_{\pi/8}(x_{2q}, x_q, x_p, x'_p{=}0) = (-x_{2q} + 2 x_q )^{3} -x_{2q}^{3} + 2 \times 3 (x_{2q} - x_q) x_p\).

\(\pi/8\) gates differ for dimension \(p=3\) and \(p>3\)~\cite{Campbell12,Howard12}. For \(p=3\), with \(\zeta = e^{2 \pi i/9}\), 
\begin{equation}
  U_v = \sum_{j=1}^{2} \zeta^{v_j} \ketbra{j}{j}
\end{equation}
where
\begin{equation}
  v = (0, 6 z' + 2 \gamma' + 3 \epsilon', 6 z' + \gamma' + 6 \epsilon') \mod 9.
\end{equation}

Despite these differences, these \(\pi/8\) gates are all diagonal with entries that are rational powers of \(e^{2\pi i/p}\). This means that their Weyl symbol actions \(S_{\pi/8}(\bs x) \notin \mathbb Z[\bs x]\) but lie in \(\mathbb Z[\bs x]\) when the dimension \(p^h\) of the system is increased to \(p^{h'}\).

Here we will show how to perform this change in the domain of summation for a gate from \(\mathbb Z / p^h \mathbb Z\) to \(\mathbb Z/ p^{h'} \mathbb Z\), where \(h'>h\), such that the resultant equivalent action will go from having coefficients in \(\mathbb Q\) to \(\mathbb Z\). This is a useful technique for general diagonal gates, since the coefficients in their action polynomials will generally fall in \(\mathbb Q\) instead of \(\mathbb Z\). Reexpressing their Weyl symbols in this way will allow for their evaluation by the method of stationary phase. Though we will only show this for the qutrit \(\pi/8\) gate as an example, this technique can be equivalently applied for all diagonal gates with rational coefficients, as well as their Clifford transformations, which only symplectically permute the action.

To demonstrate this we examine a particular \(\pi/8\) gate for qutrits and its corresponding magic state.

Let \(z'= 1\), \(\gamma' = 2\) and \(\epsilon' = 0\) so that
\begin{equation}
  U_v(0,1,8) = \left( \begin{array}{ccc} 1 & 0 & 0\\ 0 & \zeta & 0\\ 0 & 0 & \zeta^8\end{array} \right).
\end{equation}

The corresponding Weyl symbol is
\begin{eqnarray}
  U_{\pi/8}(x_{q}) &=& \exp \left[ - \frac{2 \pi i}{3} \frac{2}{3} (x_{q} + 3 x_{q}^2) \right] \\
                 &\equiv& \exp\left[\frac{2 \pi i}{3} S_{\pi/8}(x_{q})\right], \nonumber
\end{eqnarray}
for \(x_q \in \mathbb Z/ 3^2 \mathbb Z\) and \(x_{q_1} \in \mathbb Z/ 3 \mathbb Z\).

\(S_{\pi/8}(\bs x) \notin \mathbb Z[\bs x]\) but it does lie in \(\mathbb Z[\bs x]\) when the domain is increased from \(3\) to \(3^2\); we want to consider the sum over \(x_q\) of \(U_{\pi/8}(\bs x) = U_{\pi/8}(x_q)\) and reexpress it in terms of an exponential such that it is a sum over the larger space \(\mathbb Z/ 3^2 \mathbb Z\) of a polynomial with integer coefficients (see \ref{app:qutritpi8gate}):
\begin{eqnarray}
  && \sum_{x_q \in \mathbb Z/3 \mathbb Z} \exp\left[ \frac{2 \pi i}{3} \left(-\frac{2}{3}\right) (x_q + 3 x_q^2) \right] \\
  &=& \frac{1}{3} \sum_{x_q \in \mathbb Z/3^2 \mathbb Z} \exp \left[ \frac{2 \pi i}{3^2} x_q^{3} \right] \nonumber\\
  &\equiv& \frac{1}{3} \sum_{x_q \in \mathbb Z/3^2 \mathbb Z} \exp \left[ \frac{2 \pi i}{3^2} S'_{\pi/8}(\bs x) \right]. \nonumber
\end{eqnarray}
Notice that the first line consists of a sum of an exponentiated polynomial with coefficients in \(\mathbb Q\) while the second line is a sum over a larger space, but this time of an exponentiated polynomial with coefficients in \(\mathbb Z\). Thus, we have accomplished what we set out to do.

Since we will not be performing the sum above in isolation but with other exponentiated terms, so it is important to establish a slightly more general result:
\begin{eqnarray}
  &&\sum_{x_q \in \mathbb Z/3 \mathbb Z} \exp\left\{ \frac{2 \pi i}{3} \left[ -\frac{2}{3} (x_q + 3 x_q^2) + f(\bs x) \right] \right\} \\
  &=& \frac{1}{3} \sum_{x_q \in \mathbb Z/3^2 \mathbb Z} \exp \left\{ \frac{2 \pi i}{3^2} \left[ S'_{\pi/8}(\bs x) + 3 f(\bs x) \right] \right\}, \nonumber
\end{eqnarray}
where \(f(\bs x) \in \mathbb Z[x_1, \ldots, x_n]\).

\begin{widetext}
 Applying this identity to \(U_{\pi/8} U^*_{\pi/8}(\bs x, \bs x')\) transforms it into an expression with exponentiated polynomials over a larger domain of summation and whose coefficients are in \(\mathbb Z\). This result is now amenable for the discrete ``periodized'' stationary phase method result given by Theorem~\ref{th:statphase}. This produces (see Appendix~\ref{app:qutritpi8gate}):
 \begin{eqnarray}
   \label{eq:qutritpi8gate}
  U_{\pi/8} U^*_{\pi/8}(\bs x, \bs x') &=& \frac{1}{3^2} \sum_{\bs x_1, \bs x_2, \bs x_3 \in (\mathbb Z/3 \mathbb Z)^2} U_{\pi/8}(\bs x_1) U_{\pi/8}^*(\bs x_2) \exp\left[ \frac{2 \pi i}{3} \Delta_5(\bs x_3, \bs x, \bs x_1, \bs x', \bs x_2)\right], \\
                                       &=& \frac{1}{3} \sum_{\{\tilde x_{2q}\} \Mod 3} \exp\left[\frac{2 \pi i}{3^2} S''_{\pi/8}(\tilde x_{2q}) \right] G_0(H_{\tilde x_{2q}},3^{-1}\nabla S''_{\pi/8}(\tilde x_{2q}) ), \nonumber
\end{eqnarray}
where \(S''_{\pi/8}(\tilde x_{2q}) = 3 \left[ 2 (x_p - x'_p) x_{2q} + (x_p + 3 x'_p) x'_q - (3 x_p + x'_p) x_q  \right]\), \(H_{\tilde x_{2q}} = \left. \frac{\partial^2 S}{\partial x_\mu \partial x_\nu} \right|_{x = \tilde x_{2q}}\), and \(G_0\) is defined in Eq.~\ref{eq:Gausssum}.
\end{widetext}

We are interested in the magic state of this gate, which corresponds to it acting on \(\hat H \ket 0 = \ket{p=0}\). The Wigner function of \(\ket{p=0}\) is \(\rho'(\bs x) \equiv \frac{1}{3} \delta_{x_p, \bs 0}\). Hence,
\begin{widetext}
  \begin{eqnarray}
    \rho_{\pi/8}(\bs x) \equiv \sum_{\bs x'} U_{\pi/8}U^*_{\pi/8}(\bs x, \bs x') \rho'(\bs x') &=& \begin{cases} \frac{1}{3^2} \sum_{x_{2q} \in \mathbb Z/ 3^2 \mathbb Z} \exp \left[ \frac{2 \pi i}{3^2} S''_{\pi/8}(x_{2q}, x_q, x_p, x'_p{=}0) \right] & \text{if}\,\, x'_q = x_q \,(\hskip-3pt \Mod 3),\\ 0 & \text{otherwise} \end{cases}. \nonumber\\
    \label{eq:weylpi8gatequtrit}
                      &=& \frac{1}{3^2} \sum_{x_{2q} \in \mathbb Z/ 3^2 \mathbb Z} \exp \left\{ \frac{2 \pi i}{3^2} \left[ (-x_{2q} + 2 x_q )^{3} -x_{2q}^{3} + 2 \times 3 (x_{2q} - x_q) x_p  \right] \right\},
  \end{eqnarray}
  (where, again, details are shown in Appendix~\ref{app:qutritpi8gate}).
\end{widetext}

If we let \(S(x_p, x_q)\) be the phase, we note that \(\partder{S}{x_p} = \partder{S}{x_q} = 0 \Mod 3\) and \(\nabla^2 S \equiv H = 0 \Mod 3^0\) \(\forall \, x_p, \,x_q\). Hence, evaluation at any phase space point, or linear combination thereof, requires summation over the three reduced phase space points \(\tilde x_{2q} \in \mathbb Z/ 3 \mathbb Z\) since they are all critical points.

\subsection{Application}
We consider as an example, computation of the qutrit circuit outcome:
\begin{equation}
  \label{eq:probofapp}
  P = \Tr \left[ \bra{0} \hat U_C \hat U_{\pi/8}^{\otimes k} \hat H^{\otimes k}\ket{0}^{\otimes n} \right].
\end{equation}
This corresponds to the probability of the outcome \(\ketbra{0}{0}\) after \(k\) qutrit magic states are acted on by a random Clifford circuit \(U_C\) on \(n\ge k\) total qutrits.

A recent study introduced a method to sample this distribution using Monte Carlo methods on Wigner functions~\cite{Pashayan15}. Results were demonstrated for one to ten magic states, in a system of \(100\) qutrits, that required \(10^5\) to \(>10^8\) samples, respectively, to attain precision \((P_{\text{sampled}} - P) <10^{-2}\) with \(95\%\) confidence~\cite{Pashayan15}.

Before the Clifford gates are applied, in the Wigner picture the circuit can be described by \(k\) products of Eq.~\ref{eq:weylpi8gatequtrit} and \((n-k)\) products of \(\delta_{x_{q_j},0}\). The random Clifford gate can be described as an affine transformation by a symplectic matrix \(\bsmc M\) and vector \(\bs \alpha\)~\cite{Gross06,Kocia17}, as described in Section~\ref{subsec:GausssumsonCliffordgates}, such that
\begin{equation}
  \left( \begin{array}{c} \bs x'_p\\ \bs x'_q \end{array}\right) = \bsmc M \left( \left( \begin{array}{c} \bs x_p\\ \bs x_q \end{array}\right) + \frac{\bs \alpha}{2} \right) + \frac{\bs \alpha}{2}.
\end{equation}
The form of the \(\bsmc M\) and \(\bs \alpha\) for the Clifford gates are given in~\cite{Kocia17}. The final contraction to \(\ketbra{0}{0}\) corresponds to a sum over all \(n\) of the \(x_{p_j}\) and \(x_{q_j}\) except for the first degree of freedom, which is only considered at \(x_{q_1} = 0\) (and all \(x_{p_1}\)).

We note that this is a strong simulation algorithm. Therefore, it is interesting to compare the number of terms produced by the stationary phase method that must be summed over in this strong simulation with the number of samples that must be taken in the prior strong simulation by Pashayan \emph{et al}. However, a direct comparison of the two on an equal footing is a bit blurred by the fact that the latter only reports results based on a Monte Carlo sampling of their terms, producing a result correct only up to a precision \(\epsilon = 0.01\). Their explicit evaluation would produce a result correct up to precision \(\epsilon = 0.0\) and would likely scale far worse (we estimate at least \(3^{2t}\) for a naive evaluation of the \(2t\) dimensional Wigner function, see Fig.~\ref{fig:pi8gatenumerics}). Monte Carlo calculation of finite sums reduce the number of terms that must be evaluated for intermediate to large sums, and we assume that is the case with Pashayan \emph{et al}.'s results. Thus we will proceed to compare the two methods with the caveat that Monte Carlo sampling of the sums produced by the stationary phase method will likely also improve its performance for intermediate to large values of \(t\).

For \(k\) magic states in a \(100\)-qutrit system, we have \(k\) one-dimensional sums (integrals) over \(x_{2q_j}\) and \(100\) sums over all the \(x_{p_j}\) and \(99\) sums over all the \(x_{q_j}\) except \(x_{q_1}\), resulting in a total of \((k+199)\) sums. This can be concisely written by letting \(\bs x \equiv (x_{p_1}, \ldots, x_{p_{100}}, x_{q_1}, \ldots, x_{q_{100}}) \equiv (\bs x_p, \bs x_q)\) and \(\rho_{\pi/8}(\bs x_i) \equiv \rho_{\pi/8}(x_{p_i}, x_{q_i})\) such that
\begin{equation}
  P = \sum_{\bs x' \in D} \left[ \prod_{i=1}^t \rho_{\pi/8}(\bs x'_i) \prod_{j=t+1}^{100} \delta(x'_{q_j})\right],
  \label{eq:applicationtrace}
\end{equation}
for
\begin{equation}
  D = \left\{\bs x' \bigg| \left(\bsmc M^{-1} \left(\bs x'-\frac{\bs \alpha}{2}\right)-\frac{\bs \alpha}{2}\right)_{101} \Mod 3^2 =0\right\}
  \label{eq:applicationtracerestriction}
\end{equation}
where \((\bsmc M^{-1} (\bs x'-\frac{\bs \alpha}{2})-\frac{\bs \alpha}{2})_i = x_i\) for \(i \in \{1, \ldots, 200\}\). There are \(199\) sums in Eq.~\ref{eq:applicationtrace} and the \(\rho_{\pi/8}\) contain \(k\) more, for a total of \((k+199)\).

The restriction on the sum over \(\bs x'\) to be in \(D\), defined in Eq.~\ref{eq:applicationtracerestriction}, can be treated as an additional Kronecker delta function modulo \(3^2\), \(\delta(\bsmc M^{-1} (\bs x'-\frac{\bs \alpha}{2}) - \frac{\bs \alpha}{2})_{101})\), multiplying the full sum over \(\bs x'\).

We can begin reducing the number of sums by first summing out all \(\delta(x'_{q_j})\) in Eq.~\ref{eq:applicationtrace}, which will replace all \(x'_{q_j}\) in this additional Kronecker delta function for \(j \in \{t+1, \ldots, 100\}\). We can then proceed to sum away all \(x'_{p_j}\) for \(j \in \{t+1, \ldots, 100\}\) since they are not present anywhere in the full summand and put in the appropriate factors of \(3\). The additional delta function \(\delta(\bsmc M^{-1} (\bs x'-\frac{\bs \alpha}{2}) - \frac{\bs \alpha}{2})_{101})\) can now only include terms \(x'_{p_i}\) and \(x'_{q_i}\) for \(i \in \{1, \ldots, t\}\). Any such \(x'_{p_i}\) term can now be chosen to be summed away, which replaces it on the corresponding \(\rho_{\pi/8}\) state. The remaining \((t-1)\) \(x'_{p_i}\) terms can only be linear terms in the \(\rho_{\pi/8}\) exponential functions. Thus, they can be summed away to produce \((t-1)\) Kronecker delta functions that are linearly independent and so can be used to sum away \((t-1)\) \(x'_{q_i}\). There will remain left over one \(x'_{q_i}\) variable along with the \(t\) \(x_{2qi}\) variables---a remaining total of \((t+1)\) sums. This process is illustrated schematically in Table~\ref{tab:pi8gatesimplification}.

  \begin{table*}[ht]
    \begin{tabular}{| c | c | c |}
      \hline
    Step \# & Description & Resultant Equation\\
    \hline
    \(0\) & & \(\sum \left[ \prod \rho^{\text{linear in \(x'_{p_i}\)}}_{\pi/8}(x'_{p_i}, x'_{q_i}, x_{2q_i}) \rho^{\text{rest}}_{\pi/8}(x'_{q_i}, x_{2q_i}) \prod \delta(x'_{q_j}) \delta((\mathcal M \bs x')_{101})\right]\)\\
    \hline
    \(1\) & Sum away all \(x'_{q_j}\). & \(\sum \left[ \prod \rho^{\text{linear in \(x'_{p_i}\)}}_{\pi/8}(x'_{p_i}, x'_{q_i}, x_{2q_i}) \rho^{\text{rest}}_{\pi/8}(x'_{q_i}, x_{2q_i}) \delta((\mathcal M \bs x')_{101}|_{\wedge x'_{q_j}=0})\right]\)\\
    \hline
    \(2\) & Sum away all \(x'_{p_j}\). & \(3^{100-t}\sum \left[ \prod \rho^{\text{linear in \(x'_{p_i}\)}}_{\pi/8}(x'_{p_i}, x'_{q_i}, x_{2q_i}) \rho^{\text{rest}}_{\pi/8}(x'_{q_i}, x_{2q_i}) \delta((\mathcal M \bs x')_{101}|_{\wedge x'_{q_j}=0})\right]\)\\
    \hline
    \(3\) & \begin{tabular}{c}Replace a \(x'_{p_{i'}}\) term\\ using remaining delta function.\end{tabular} & \(\begin{array}{c}3^{100-t}\sum \left[ \prod_{i\ne i'} \rho^{\text{linear in \(x'_{p_i}\)}}_{\pi/8}(x'_{p_i}, x'_{q_i}, x_{2q_i}) \rho^{\text{rest}}_{\pi/8}(x'_{q_i}, x_{2q_i})\right. \\ \times \left.\rho^{\text{linear in \(x'_{p_{i'}}\)}}_{\pi/8}(\bs x', x_{2q_{i'}}) \rho^{\text{rest}}_{\pi/8}(x'_{q_{i'}}, x_{2q_{i'}}) \right]\end{array}\)\\
    \hline
    \(4\) & Sum away remaining \((t-1)\) \(x'_{p_i}\) terms. & \(\begin{array}{c}3^{100-t}\sum \left[ \prod_{i\ne i'} \delta(x'_{p_i}, x'_{q_i}, x_{2q_i}) \rho^{\text{rest}}_{\pi/8}(x'_{q_i}, x_{2q_i})\right. \\ \times \left.\rho^{\text{linear in \(x'_{p_{i'}}\)}}_{\pi/8}(\bs x'_q, x_{2q_{i'}}) \rho^{\text{rest}}_{\pi/8}(x'_{q_{i'}}, x_{2q_{i'}}) \right]\end{array}\)\\
    \hline
    \(5\) & Sum away \((t-1)\) \(x'_{q_i}\). & \(\begin{array}{c}3^{100-t}\sum \left[ \prod_{i\ne i'} \rho^{\text{rest}}_{\pi/8}(x'_{q_{i'}}, x_{2q_i})\right. \\ \times \left.\rho^{\text{linear in \(x'_{p_{i'}}\)}}_{\pi/8}(x'_{q_{i'}}, x_{2q_{i'}}) \rho^{\text{rest}}_{\pi/8}(x'_{q_{i'}}, x_{2q_{i'}}) \right]\end{array}\)\\
    \hline
    \end{tabular}
  \caption{Step-by-step schematic illustration of the worst-case simplification of Eq.~\ref{eq:applicationtrace} to \((t+1)\) sums described in the text. For clarity, in the equations in the third column, \(\rho_{\pi/8}(x'_{p_i}, x'_{q_i}, x_{2q_i})\) is split up into its exponential part linear in \(x'_{p_i}\), \(\rho^{\text{linear in \(x'_{p_i}\)}}(x'_{p_i}, x'_{q_i}, x_{2q_i})\), and the rest, \(\rho^{\text{rest}}_{\pi/8}(x'_{q_i}, x_{2q_i})\), which is not dependent on \(x'_{p_i}\). Furthermore, Kronecker delta functions are mostly written in short-hand in terms of just their comma-separated arguments. It should be possible from the text to determine the domains of the sums and products as well as the form of the functions listed.}
    \label{tab:pi8gatesimplification}
\end{table*}

If the additional Kronecker delta function representing the restricted sum never had any remaining \(x'_{p_i}\) terms then it could simply be used to sum away and replace a \(x'_{q_i}\) term with other \(x'_{q_i}\)s, and one could then proceed as outlined before by summing away all \(x'_{p_i}\) to produce Kronecker delta functions and summing away all \(x'_{q_i}\). This would result in only \(t\) remaining sums over the \(x_{2q_i}\).

Therefore, in the worst case, the total number of variables that will need to be summed in the \(t\) \(\pi/8\) gate magic states will be \((t+1)\) consisting of Gauss sums and leading phases. This is worst-case scaling of \(3^{t+1}\) terms in a sum. We show how this scaling compares to the alternative algorithm in Fig.~\ref{fig:pi8gatenumerics}.

Acting \(\hat U_C\) on the ket instead produces \(\bsmc M (\bs x + \frac{\bs \alpha}{2}) + \frac{\bs \alpha}{2}\) in the arguments for the initial stabilizers and \(\rho_{\pi/8}\) functions with a sum over \(x\) restricting \(x_{q_1} = 0\). The preceding argument then can be made again since the rows in the top half and bottom half of \(\bsmc M (\bs x + \frac{\bs \alpha}{2}) + \frac{\bs \alpha}{2}\) correspond to linearly independent equations and so can be treated independently~\cite{Kocia17}. This results in a worst-case of \((t+1)\) sums again.

\begin{figure}[ht]
  \includegraphics[scale=1.0]{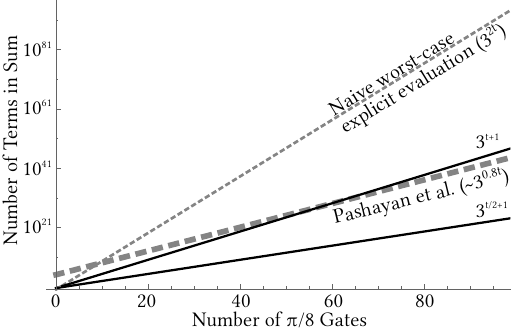}
  \caption{Logarithm of the number of terms required to evaluate single qutrit marginal after Clifford evolution of Eq.~\ref{eq:weylpi8gatequtrit}. The number of samples needed for a full Wigner calculation (thinner dashed curve) and the corresponding Monte Carlo integration based on negativity for precision \(<0.01\) from Pashayan \emph{et al}.~\cite{Pashayan15} (thicker dashed curve) is included for reference. The upper black curve (plotting \(3^{t+1}\)) illustrates the scaling of the number of terms in the sum with the single qutrit \(\pi/8\) gate magic states (Eq.~\ref{eq:qutritpi8gate}) and the lower black curve (plotting \(3^{\frac{t}{2}+1}\)) illustrates the scaling with the two-qutrit \(\pi/8\) gate magic state (Eq.~\ref{eq:twoqutritpi8gate_2}). The latter illustrates that the stationary phase algorithm for strong simulation (full calculation, i.e. with precision \(0.0\)) is able to significantly outperform Pashayan \emph{et al}.'s algorithm based on negativity for any number of qutrits with a performance gain that grows exponentially with the number of qutrits.}
  \label{fig:pi8gatenumerics}
\end{figure}

We have thus found that a \(\pi/8\) gate (and magic state) can be captured with only \(p\) critical points. This is commensurate with its role as the ``minimally'' non-Clifford gate~\cite{Boykin99} and that its direct products have the smallest stabilizer rank~\cite{Bravyi16}.

We proceed to improve this scaling further by developing a relationship between state stabilizer rank and the number of critical points necessary to represent that state and then leveraging this result with the \(\pi/8\) gate magic state's optimal stabilizer rank for two qutrits.

\section{Stabilizer Rank and Stationary Phase Critical Points}
\label{sec:stabrank}

We will now see that a relationship can be established between the stabilizer rank of this state and the number of these critical points. This is in contrast to the relationship between the amount of negativity present magic states and contextuality, which appear to be inversely related~\cite{Howard17}. This suggests that negativity is not the most efficient way to introduce ``magic'' or non-contextuality in a practical algorithm, and indeed we find this to be the case here.

In the stationary phase method applied to the infinite dimensional (continuous) case, the critical points correspond to intersections between Gaussian manifolds, the continuous generalization of stabilizer states~\cite{Kocia16}. However, in the ``periodized'' stationary phase method examined here, we cannot expect the same relationship to hold. Nevertheless, a relationship can still be estabilished between the number of stabilizer states that a state can be expressed in terms of---its stabilizer rank---and the number of critical points necessary to represent that state:
\begin{theorem}[Stabilizer Rank and Critical Points]
  \label{th:stabrank}
  Let \(\ket{\Psi}\) be a \(n\)-qudit odd-prime-\(p\)-dimensional state that can be written as an equiprobable linear combination of the \(p^n\) logical basis states. If \(\ket{\Psi}\) can also be expressed as an equiprobable linear combination of \(p^m\) (\(m \le n\)) orthogonal stabilizer states, which can be written, after some Clifford transformation \(\hat U_C\), as products of \(m\) orthogonal single-qudit stabilizer states \(\{\ket{\phi_{ij}}\}\) and \((n{-}m)\) single-qudit stabilizer states \(\{\ket{\psi_{ik}}\}\) that are not logical states,
  \begin{equation}
    \hat U_C \ket{\Psi} = \sum_{i=0}^{p^m-1} c_i \prod_j^{m} \ket{\phi_{ij}} \prod_k^{n-m} \ket{\psi_{ik}},
  \end{equation}
  where \(c_i \in \mathbb C\) and \(|c_i|^2 = |c_j|^2\) \(\forall\, i,j\), then the Wigner function of \(\hat \rho = \ketbra{\Psi}{\Psi}\), \(\rho(\bs x)\), can be expressed in terms of \(\le p^m\) quadratic Gauss sums of dimension \(p^{n-m}\) (i.e. \(\rho(\bs x)\) can be written in terms of \(\le p^m\) critical points of dimension \(p^{n-m}\)).
\end{theorem}

\begin{proof}
  For a fixed \(j\), there must exist a one-qubit Clifford transformation that takes \(\ket{\phi_{ij}}\) to \(\ket{0}\). It follows after this transformation, for all other \(i\), \(\ket{\phi_{ij}}\) are also orthogonal one-qudit stabilizer states and so must be the other logical states on that qudit. Proceeding in this manner on the remaining one-qudit \(\{\phi_{ij}\}\) produces a Clifford transformation \(\hat U'_C\) such that
  \begin{equation}
    \hat U'_C \hat U_C \ket{\Psi} = \sum_i^{p^m} c_i \ket{i} \prod_k^{m-n} \ket{\psi_{ik}},
  \end{equation}
  where \(\ket{i}\) is the base \(p\) representation of a product of single-qudit logical (stabilizer) states.

  We can rewrite the \(({n-m})\) remaining qudits in the base \(p\) representation:
  \begin{equation}
    \sum_i^{p^m} c_i \ket{i} \prod_k^{n-m} \ket{\psi_{ik}}, = \sum_i^{p^m} \sum_j^{p^{n-m}} c_{i,j} \ket{i} \ket{j},
  \end{equation}
  where again \(\ket{j}\) is written in base \(p\) representation.

  \(|c_{i,j}|^2 = |c_{k,l}|^2\) \(\forall\, i,j,k,l\) since we are given that \(|c_i|^2 = |c_j|^2\) \(\forall\, i,j\) and \(\{\ket{\psi_{ik}}\}\) are single-qudit stabilizer states that are not logical states, and so are equiprobable in terms of logical states.
  
  It follows that
  \begin{equation}
    c_{i,j} = \mathcal N \exp \left( \frac{2 \pi i}{p} S(i,j) \right),
  \end{equation}
  where \(S(i,j) \in \mathbb R\) is some polynomial function of \(i\) and \(j\) when they are broken up into their \(m\) and \((n-m)\) base \(p\) constituents, respectively.

  Therefore, the Weyl symbol of the diagonal gate \(\hat U\) that takes \(\hat \rho' = \prod_{i=1}^{p^n} \hat F_i \ket{0}^{\otimes n}\) to \(\sum_i^{p^m} c_i \ket{i} \ket{\psi_{i}}\), \(\hat U = \sum_i^{p^m} \sum_j^{p^{n-m}} c_{i,j} \ketbra{i,j}{i,j}\), is \(U(i,j) = \mathcal N \exp \left( \frac{2 \pi i}{p} S(i,j) \right)\).

  For fixed \(i\), \(\sum_j c_{i,j} \ketbra{i}{i} \otimes \ketbra{j}{j}\) takes \(\ket{i} \prod_{i=1}^{p^{n-m}} \hat F_i \ket{0}^{\otimes (n-m)}\) to \(\sum_j^{p^{n-m}} c_{i,j} \ket{i} \ket{j} = \ket{i} \ket{\psi_i}\), and is thus a Clifford gate since it takes a stabilizer state to a stabilizer state. Hence, for fixed \(i\), \(S(i,j)\) must be a quadratic polynomial in \(j\).

  Therefore, each \(i\) in the sum indexes a term made up of a \(p^{n-m}\)-dimensional quadratic Gauss sum. This implies that the Wigner function has \(\le p^m\) quadratic Gauss sums of dimension \(p^{n-m}\), or equivalently, critical points. \qed
\end{proof}

\section{Two Qutrit \(\pi/8\) gate magic States}
\label{sec:twoqutritmagicstates}

A single qubit \(T\)-gate magic state is not a stabilizer state and has a stabilizer rank of two. Two copies of the qubit magic state are somewhat remarkable in that they can also be written in terms of only two stabilizer states:
\begin{eqnarray}
  \ket{A^{\otimes 2}} &=& \frac{1}{2} (\ket{00} + i \ket{11})\nonumber\\
                      &&+\frac{e^{i\pi/4}}{2}(\ket{01}+\ket{10}),
\end{eqnarray}
where each line consists of a stabilizer state.

In the qutrit case, a similar pattern holds. A single qutrit \(T\)-gate magic state consists of at least three stabilizer states. Moreover, two copies of the qutrit magic state can also be written in terms of only three stabilizer states:
\begin{eqnarray}
  \label{eq:pi8gate3stabstates}
  \hat U_{\pi/8}^{\otimes 2} \hat F^{\otimes 2} \ket{00} &=& \frac{1}{3}(\ket{00} + \ket{12} + \ket{21})\\
                                                         &&+ \frac{e^{-2 \pi i/9}}{3} (\ket{02} + \ket{20} + e^{2 \pi i/3} \ket{11}) \nonumber\\
                                                         && + \frac{e^{2 \pi i/9}}{3} (\ket{01} + \ket{10} + e^{-2 \pi i/3}\ket{22}), \nonumber
\end{eqnarray}
where again each line consists of a stabilizer state (\(\hat F\) is the Hadamard gate or discrete Fourier transform).

Qubit strong simulation algorithms based on stabilizer rank have been able to leverage this fact to halve their exponential scaling of terms to \(\mathcal O(2^{0.5 t})\)~\cite{Bravyi16,Howard18}. The performance can be improved to \(\mathcal O (2^{\sim 0.468 t})\) by using the fact that six qubit \(T\)-gate magic states can be written in terms of only seven stabilizer states~\cite{Bravyi16_2}.

The dependence of the number of critical points on the intermediate values of \(\bs x_{2q}\) for the \(\pi/8\) gate magic state, means that generally the number of critical points scales exponentially with number of \(\pi/8\) gates. However, if this stabilizer rank for two qutrit \(\pi/8\) gate magic states can be used, the exponent is reduced by a factor of \(2\).

Eq.~\ref{eq:pi8gate3stabstates} can be transformed into a more presentable form by the action of the Clifford controlled-not gate, \(C_{12}\), which adds the value of the low qutrit to the high qutrit mod $3$ and produces:
 \begin{eqnarray}
C_{12} \hat U_{\pi/8}^{\otimes 2} \hat F^{\otimes 2} \ket{00} &=&\frac{1}{3} \left( \ket{00} + \ket{02} + \ket{01} \right)\\
&&+ \frac{e^{-2 \pi i/9}}{3} \left( \ket{22} + \ket{20} + e^{2 \pi i/3} \ket{21} \right) \nonumber\\
&& + \frac{e^{2 \pi i/9}}{3} \left( \ket{11} + \ket{10} + e^{-2 \pi i/3}\ket{12} \right), \nonumber
 \end{eqnarray}
 We recognize this as a pointer state, in which the high qutrit tells us the state of the low qubit. These states are stabilizer states, and we may write:
\begin{eqnarray}
C_{12} \hat U_{\pi/8}^{\otimes 2} \hat F^{\otimes 2} \ket{00} &=& \frac{1}{3} \ket{0} \hat F \ket{0} \\
&+& \frac{e^{2 \pi i/9}}{3} \ket{1} \hat P^2 \hat F \ket{0} \nonumber\\
&+& \frac{e^{-2 \pi i/9}}{3}\ket{2} \hat P \hat F \ket{1}, \nonumber
\end{eqnarray}
where \(\hat P = \text{diag} \{1,1,\omega\}\) is the Clifford phase shift gate.

This form satisfies the description of Theorem~\ref{th:stabrank} and so implies that the Wigner function of \(\hat U_{\pi/8}^{\otimes 2} \hat F^{\otimes 2} \ket{00}\) can be represented by three one-dimensional quadratic Gauss sums, just like for \(\hat U_{\pi/8} \hat F \ket{0}\)---the single \(\pi/8\) gate magic state.

To see how this is true in the Wigner representation, we consider two copies of the magic state given by Eq.~\ref{eq:weylpi8gatequtrit}:
\begin{widetext}
  \begin{eqnarray}
    \rho^{\otimes 2}_{\pi/8}(\bs x_1, \bs x_2) &=& \frac{1}{3^2} \sum_{x_{2q_1} \in \mathbb Z/ 3^2 \mathbb Z} \exp \left\{ \frac{2 \pi i}{3^2} \left[ (-x_{2q_1} + 2 x_{q_1} )^{3} - x_{2q_1}^{3} + 2 \times 3 (x_{2q_1} - x_{q_1}) x_{p_1}  \right] \right\} \\
                                                && \times \frac{1}{3^2} \sum_{x_{2q_2} \in \mathbb Z/ 3^2 \mathbb Z} \exp \left\{ \frac{2 \pi i}{3^2} \left[ (-x_{2q_2} + 2 x_{q2} )^{3} - x_{2q_2}^{3} + 2 \times 3 (x_{2q_2} - x_{q_2}) x_{p_2}  \right] \right\} \nonumber
  \end{eqnarray}

We are allowed to transform \(x_{2q_1}\) and \(x_{2q_2}\) by a symplectic transformation since the sum over these variables is invariant to symplectic transformations; the sum over a finite field of an exponential polynomial function that is linearly transformed merely changes the order of summation. However, to calculate the full trace as in Eq.~\ref{eq:applicationtrace} involves summing over \(\bs x\) as well, which involves terms that are also exponential polynomial functions but over a subset of the full domain \(D\) given by Eq.~\ref{eq:applicationtracerestriction}. A symplectic transformation does not leave this restricted sum invariant, it additionally transforms the restriction (Eq.~\ref{eq:applicationtracerestriction}) by some matrix \(\bsmc M\) and vector \(\bs \alpha\). Fortunately, we consider all such transformations of the restriction in our analysis above due to already considering any Clifford transformation on the \(\pi/8\) gates. Therefore, we are free to further symplectically transform the \(\bs x\) variables within our analysis.

   The two-qutrit \(\hat C_{12}\) controlled-not gate transforms \((x_{p_1}, x_{p_2}, x_{q_1}, x_{q_2})\) to \((x_{p_1}, x_{p_2}-x_{p_1}\Mod 3, x_{q_1}+x_{q_2}\Mod 3, x_{q_2})\)~\cite{Kocia16}. Acting on both \(\bs x\) and \(x_{2q_1}\) and \(x_{2q_2}\) with \((\hat C_{12})^2\) produces:
  \begin{eqnarray}
    \label{eq:twoqutritpi8gate_2}
    && \rho_{\pi/8}^{\otimes 2}(x_{p_1}, x_{p_2} + x_{p_1} \Mod 3, x_{q_1} - x_{q_2} \Mod 3, x_{q_2}) \\
    &=& \frac{1}{3^2} \sum_{x_{2q_1}, x_{2q_2} \in \mathbb Z/ 3 \mathbb Z} \exp \left\{ \frac{2 \pi i}{3^2} \left[ 8 x_{q_1}^3 + 7 x_{2q_1}^3 + P(x_{2q_1}, x_{2q_2}, \bs x) \right] \right\},
\end{eqnarray}
where
\begin{eqnarray}
  P(x_{2q_1}, x_{2q_2}, \bs x) &=& 3 x_{q_1}^2 x_{q_2} + 6 x_{q_1} x_{q_2}^2 + 6 x_{q_1}^2 x_{2q_1} + 6 x_{q_1} x_{q_2} x_{2q_1} \nonumber\\
    && + 6 x_{q_2}^2 x_{2q_1} + 6 x_{q_1} x_{2q_1}^2 + 3 x_{q_2} x_{2q_1}^2 + 3 x_{q_1}^2 x_{2q_2} + 3 x_{q_1} x_{q_2} x_{2q_2} \nonumber\\
    && + 6 x_{q_1} x_{2q_1} x_{2q_2} + 3 x_{q_2} x_{2q_1} x_{2q_2} + 6 x_{2q_1}^2 x_{2q_2} + 6 x_{q_1} x_{2q_2}^2 + 3 x_{2q_1} x_{2q_2}^2  \nonumber\\
    && + x_{p_1} (6 x_{2q_1} + 3 x_{q_1}) + x_{p_2} (6 x_{2q_2} + 3 x_{q_2}). \nonumber
\end{eqnarray}
\end{widetext}

We can see that for each of the three values that \(x_{2q_1}\) takes in the sum, the exponent is a quadratic polynomial for \(x_{2q_2}\) for fixed \(\bs x\) with coefficients in \(\mathbb Z/3 \mathbb Z\) and so is a quadratic Gauss sum. We further see that the \(x_{q_1}\) are similarly cubic and that for each value of \(x_{q_1}\) the exponent is a quadratic polynomial for \(x_{q_2}\) for fixed \(x_{2q_1}\) and \(x_{2q_2}\). \(x_{p_1}\) and \(x_{p_2}\) are both linear terms just as before, and therefore the same arguments for the scaling of the marginal trace hold here.

Thus, the end result is a sum of three one-dimensional Gauss sums, exactly the same number of Gauss sums as we found for a single \(\pi/8\) gate in Section~\ref{sec:qutritpi8gate}.

Thus the number of critical points is reduced to \(3^{t/2+1}\) so that the cost of simulation of the example in Section~\ref{sec:qutritpi8gate} is \(\mathcal O(3^{t/2})\). This can be seen by the lowest black curve in Figure~\ref{fig:pi8gatenumerics}.

This performance is better than the Monte Carlo algorithm of Pashayan~\emph{et al}.~\cite{Pashayan15} and is achieved for an algorithm with no Monte Carlo sampling error. Use of Monte Carlo would further improve this scaling. It is interesting to note that weak simulation of qutrits by the method of~\cite{Huang18} scales as \(3^{0.32t}\) and so it is perhaps possible that Monte Carlo could improve this strong simulation algorithm to be more efficient than weak simulation.
  
\section{Future Directions}
\label{sec:future}

One of the central pillars of the stationary phase method used here is that non-contextual operations are efficiently classically simulable. This means that the Weyl symbols of Clifford gates reduce to a Gauss sum and affect simple symplectic transformations in Weyl phase space. Furthermore, stabilizer states have non-negative Wigner functions that characterize functions of affine subspaces of the discrete phase space.

For a similar approach to work for qubits, the same operations must be non-contextual in order to be able to be free resources. This requires the WWM formalism to be extended from two generators, \(p\) and \(q\), to three generators that become Grassmann elements~\cite{Kocia17_2}. The resultant Grassmann algebra cannot be treated over disjoint states in phase space~\cite{Kocia18}, such as we have done here, and so a Grassmann calculus must be used. It would be interesting to see if such the stationary phase method can be applied to the Grassmann algebra and produce a similar treatment of qubit \(\pi/8\) gates.

Another interesting direction for future study regards the mathematical relationship between \(\pi/8\) gate magic state stabilizer rank and the number of critical points in exponentiated multidimensional polynomials. The improvement in scaling found in Section~\ref{sec:twoqutritmagicstates} found by a symplectic transformation of the \(x_{2q_1}\)-\(x_{2q_2}\) plane is really due to the reduction in the cubic power of a polynomial with respect to one degree when it is reexpressed as an inseparable polynomial with the second degree of freedom. Theorem~\ref{th:stabrank} strongly suggests that a similar simplification holds for higher numbers of \(\pi/8\) gate magic states that are known to have lower stabilizer ranks~\cite{Bravyi16}. The form of this relationship may be helpful in finding such lower stabilizer ranks for even higher numbers of \(\pi/8\) gate magic states than are currently know, as well as establishing concrete bounds.

A related interesting question regards approximate stabilizer rank of magic states instead of their exact stabilizer rank. This is sufficient for weak simulation and frequently leads to a more efficient algorithm since we can get away with introducing error in the probability distribution we are sampling that is supposed to represent the exact probability distribution. A weak simulation result for qudits has recently been developed~\cite{Huang18}. However, it would be interesting to see if there is a similar result to Theorem~\ref{th:stabrank} that deals with approximate stabilizer rank and if there is some such ``approximate'' analog to the discrete stationary phase method that is useful for weak simulation.

\section{Conclusion}
\label{sec:conc}

This paper established the stationary phase method as a way to understand the order \(\hbar^0\) non-contextual Clifford subtheory in the WWM formalism, producing single Gauss sums, and higher order \(\hbar\) contextual extensions of the subtheory, in terms of uniformizations---higher order sums---that can be reexpressed in terms of a sum over critical points or Gauss sums. This firmly tracks with the same relationships that exist in the continuous infinite-dimensional Hilbert space treatment of Gaussianity and non-Gaussianity, even though the stationary phase method introduced here for qudit systems differs in that it is a ``periodized'' stationary phase.

We discussed these differences and similarities between the continuous and discrete case. This involved comparing this measure of non-contextuality to negativity. We found that the usage of higher order \(\hbar\) uniformizations through the stationary phase method is more efficient than using negativity for \(\pi/8\) gate magic states, at least in the manner that has so far been tried. This also seems to have been noticed in the discrete community, which has turned to favor stabilizer state decomposition of magic states~\cite{Bravyi16,Howard18}. By relating the stabilizer rank of magic states to the number of critical points necessary to treat them, this paper falls in line with this latter approach in treating non-Clifford gates.

We found that we are able to calculate a single qutrit marginal from a system consisting of \(t\) \(\pi/8\) gate magic states that are then evolved under Clifford gates, with a sum consisting of \(3^{t+1}\) critical points corresponding to closed-form Gauss sums when the magic states are kept separable. This scaling improves to \(3^{\frac{t}{2}+1}\) when pairs of magic states are rotated into each other in accord with the optimal two-qutrit \(\pi/8\) gate magic state stabilizer rank. We showed that the latter scaling improves upon the current state-of-the-art.

All of this taken together establishes the usefulness of contextuality for practical application of classical simulation of qudit quantum algorithms through the venue of semiclassical higher order corrections in \(\hbar\) accomplished by the stationary phase method.

\acknowledgments
Parts of this manuscript are a contribution of NIST, an agency of the US government, and are not subject to US copyright. PSL acknowledges support from NSF award PHY1720395.

\bibliography{biblio}{}
\bibliographystyle{unsrtnat}

\appendix

\section{\(p\)-adic Numbers}
\label{app:padicnumbers}

For prime \(p\), we define the \(p\)-adic order or \(p\)-adic valuation, \(\nu_p(n)\), of a non-zero integer \(n\) to be the highest exponent \(\nu\) such that \(p^\nu\) divides \(n\) and set \(\nu_p(0) = \infty\).

\(p\)-adic numbers \(\mathbb Q_p\) can be written as:
\begin{equation}
  \label{eq:padicnumber}
  \sum_{i=k}^\infty c_i p^i,
\end{equation}
for \(c_i \in \{0, 1, 2, \ldots, (p-1)\}\) and \(k\) is an integer. \(p\)-adic integers, \(\mathbb Z_p\), are \(p\)-adic numbers where \(c_i=0\) for all \(i<0\).

Let us define the absolute value of a \(p\)-adic number \(n\), \(|n|_p\), to be the inverse of \(p\) taken to the power of its valuation: \(|n|_p = \frac{1}{p^{\nu_p(n)}}\). Hence, we can define the metric \(|n-m|_p\) to denote the distance between two \(p\)-adic numbers \(m\) and \(n\). Notice that this means that \(m\) and \(n\) are ``close'' together if their distance is a \emph{large} power of \(p\), which is the opposite expected from the Euclidean metric. As a result, this \(p\)-adic formalism presents an alternative way to complete the rational numbers to the real numbers \(\mathbb R\); completing the rationals with respect to the \(p\)-adic metric, \((\mathbb Q,|\bullet|_p)\), produces the \(p\)-adic numbers.

The results presented in this paper come from \(p\)-adic number theory. However, it should be clear from this brief presentation that positive integers and positive rational numbers with terminating base \(p\) expansions will have terminating \(p\)-adic expressions (Eq.~\ref{eq:padicnumber}) that are identical to their base \(p\) expansions. Since these are primarily the cases we will be concerned with in this paper, a more thorough understanding of \(p\)-adic theory is not really strictly necessary.

\section{Gauss Sums on Clifford Gates}
\label{app:Cliffgate}
We define the Hessian
\begin{equation}
  \bs A = - 2 \left( \begin{array}{cc}- \bs B \Mod p & \bsmc J\\ -\bsmc J &\bs B \Mod p\end{array} \right),
\end{equation}
the vector
\begin{equation}
  \bs v(\bs x, \bs x', \bs x_3) = \bsmc J \left( \begin{array}{c} 2(\bs x_3 - \bs x + \bs x') + (\bs \alpha \Mod p) \\ 2(\bs x_3 - \bs x - \bs x') + (- \bs \alpha \Mod p )\end{array}\right),
\end{equation}
and the scalar
\begin{equation}
  c(\bs x, \bs x', \bs x_3) = 2 \left[\bs x'^T \bsmc J \bs x + \bs x_3^T \bsmc J (\bs x + \bs x') \right],
\end{equation}
so that the sum can be rewritten to make use of Proposition~\ref{prop:Gausssums}b:
\begin{widetext}
\begin{eqnarray}
  UU^*(\bs x, \bs x') &=& \frac{1}{p^{2n}} \sum_{\bs x_1, \bs x_2, \bs x_3 \in (\mathbb Z/p \mathbb Z)^{2n}} \exp \left\{ \frac{\pi i}{p} \left[ \left(\begin{array}{c}\bs x_1\\ \bs x_2\end{array}\right)^T \bs A \left(\begin{array}{c}\bs x_1\\ \bs x_2\end{array}\right) + 2 \bs v^T \left(\begin{array}{c}\bs x_1\\ \bs x_2\end{array}\right) + 2 c \right] \right\}\\
                       &=& \sum_{\bs x_3 \in (\mathbb Z/p \mathbb Z)^{2n}} \exp \left[ -\frac{\pi i}{p} \left( \bs u^T \bs A \bs u - 2 c \right) \right] G_1(\bs A, \bs 0), \nonumber
\end{eqnarray}
where \(\bs u = \bs A^{-1} \bs v\), for~\cite{Bernstein05}
\begin{equation}
  \bs A^{-1}= -\frac{1}{2} \left( \begin{array}{c|c} -\bs I & \bs B^{-1} \bsmc J\\\hline -\bs B^{-1} \bsmc J & \bs I \end{array} \right) (\bs B - \bsmc J \bs B^{-1} \bsmc J )^{-1} =-\frac{1}{2} (\bs B - \bsmc J \bs B^{-1} \bsmc J)^{-1} \left( \begin{array}{c|c} -\bs I & \bsmc J \bs B^{-1}\\\hline -\bsmc J \bs B^{-1} & \bs I \end{array} \right),
\end{equation}
where we assume \(\det \bs B \ne 0\). We drop the arguments on \(\bs u\), \(\bs v\), and \(c\) and such subsequent terms for conciseness.

Since \(\bs B\) can be expressed with entries in \(\mathbb Z/ p \mathbb Z\) it follows that \((\bs B - \bsmc J \bs B^{-1} \bsmc J)^{-1} \in \mathbb Z_p\) and so \(\bs u \in \mathbb Z_p\). Thus, \(G_1(\bs A,0) \ne 0 \, \forall \, \bs x, \bs x', \bs x_3 \in (\mathbb Z/ p \mathbb Z)^{2N}\) by Proposition~\ref{prop:Gausssums}b.

In preparation of the final sum over \(\bs x_3\), we expand out the exponent's phase:
  \begin{eqnarray}
  \bs u^T \bs A \bs u - 2 c &=& \bs v^T (\bs A^{-1})^T \bs A \bs A^{-1} \bs v - 2 c\\
                            &=& \tiny \left( \begin{array}{c} 2(\bs x_3 - \bs x + \bs x') + (\bs \alpha \Mod p)\\ 2(\bs x_3 - \bs x - \bs x') + (- \bs \alpha \Mod p) \end{array} \right)^T \bsmc J^T \frac{-1}{2} [(\bs B - \bsmc J \bs B^{-1} \bsmc J)^{-1}]^T \left( \begin{array}{cc} -\bs I & \bsmc J \bs B^{-1} \\ - \bsmc J \bs B^{-1} & \bs I \end{array} \right) \bsmc J \left( \begin{array}{c} 2(\bs x_3 - \bs x + \bs x') + (\bs \alpha \mod p)\\ 2(\bs x_3 - \bs x - \bs x') + (- \bs \alpha \Mod p) \end{array} \right) \nonumber\\
                            && -4 [\bs x'^T \bsmc J \bs x + \bs x_3^T \bsmc J (\bs x + \bs x')] \nonumber\\
                            &=& \left\{ [2(\bs x_3 - \bs x + \bs x') + (\bs \alpha \Mod p)]^T \bsmc J^T (\bs B - \bsmc J \bs B^{-1} \bsmc J)^{-1}\frac{-1}{2}(-\bs I) \bsmc J  \right.\nonumber\\
                            && \left. + [2(\bs x_3 - \bs x - \bs x')+(- \bs \alpha \Mod p)]^T \bsmc J^T \frac{-1}{2} (\bs B - \bsmc J \bs B^{-1} \bsmc J)^{-1}(-\bsmc J \bs B^{-1}) \bsmc J \right\} [2 (\bs x_3 - \bs x + \bs x') + (\bs \alpha \Mod p) ] \nonumber\\
                            && + \left\{ [2(\bs x_3 - \bs x + \bs x') + (\bs \alpha \Mod p)]^T \bsmc J^T \frac{-1}{2} (\bs B - \bsmc J \bs B^{-1} \bsmc J)^{-1} \bsmc J \bs B^{-1} \bsmc J \right.\nonumber\\
                            && \left. + [2(\bs x_3 - \bs x - \bs x') + (- \bs \alpha \Mod p))^T \bsmc J^T \frac{-1}{2} (\bs B - \bsmc J \bs B^{-1} \bsmc J)^{-1} \bs I \bsmc J \right\} [ 2(\bs x_3 - \bs x - \bs x') + (- \bs \alpha \Mod p)] \nonumber\\
                            && - 4 [\bs x'^T \bsmc J \bs x + \bs x_3^T \bsmc J (\bs x + \bs x')]. \nonumber
\end{eqnarray}

Gathering powers of \(\bs x_3\), we define the Hessian
\begin{eqnarray}
  \bs A' &=& 4 \bsmc J^T \frac{-1}{2} (\bs B - \bsmc J \bs B^{-1} \bsmc J)^{-1} (-\bs I - \bs B^{-1} \bsmc J) \bsmc J \\
         && 4 \bsmc J^T \frac{-1}{2} (\bs B - \bsmc J \bs B^{-1} \bsmc J)^{-1} (\bs B^{-1} \bsmc J + \bs I) \bsmc J \nonumber\\
         &=& 0, \nonumber
\end{eqnarray}
the vector
\begin{eqnarray}
  2 \bs v'^T &=& (2 \bs x' + (\bs \alpha \Mod p))^T (4 \bsmc J) (\bsmc J \bs B - \bs I)^{-1} + 4 (\bs x + \bs x')^T \bsmc J,
\end{eqnarray}
where we made use of the property that \(\bs B\), \(\bs B^{-1}\), and \((\bs B - \bsmc J \bs B^{-1} \bsmc J)^{-1}\) are symmetric, \(\bsmc J^{-1} = \bsmc J^T = - \bsmc J\), and
\begin{equation}
  \bs B - \bsmc J \bs B^{-1} \bsmc J = (\mp \bsmc J \bs B^{-1} \pm 1)(\mp \bs B \bsmc J \pm 1) \bsmc J = \bsmc J(\mp \bs B^{-1} \bsmc J \pm 1)(\mp \bsmc J \bs B \pm 1),
\end{equation}
and we define the scalar
\begin{eqnarray}
  2 c' &=& \tiny 4 \bs x'^T [4 \bsmc J^T (\bs B - \bsmc J \bs B^{-1} \bsmc J)^{-1} \frac{-1}{2} \bsmc J - 4 \bsmc J^T (\bs B - \bsmc J \bs B^{-1} \bsmc J)^{-1} \frac{-1}{2} \bsmc J \bs B^{-1} \bsmc J - \bsmc J] ((1 + \bsmc J \bs B)^{-1} (1 - \bsmc J \bs B)(\bs x' + \frac{\bs \alpha}{2}) + \frac{\bs \alpha}{2}) \nonumber\\
       && + 8 \bs \alpha^T \bsmc J^T (\bs B - \bsmc J \bs B^{-1} \bsmc J)^{-1} \frac{-1}{2} \bsmc J [1 - \bs B^{-1} \bsmc J]((1 + \bsmc J \bs B)^{-1}(1 - \bsmc J \bs B)(\bs x' + \frac{\bs \alpha}{2}) + \frac{\bs \alpha}{2}),
\end{eqnarray}
for \(\bs \beta \equiv - \frac{1}{2} \bsmc J^T (\bs B - \bsmc J \bs B^{-1} \bsmc J)^{-1} \).

In summary:
\begin{eqnarray}
  \bs A' &=& 0,
\end{eqnarray}
the vector
\begin{eqnarray}
  2 \bs v'^T &=& (2 \bs x' + (\bs \alpha \Mod p))^T (4 \bsmc J) (\bsmc J \bs B - \bs I)^{-1} + 4 (\bs x + \bs x')^T \bsmc J,
\end{eqnarray}
and
\begin{eqnarray}
  2 c' &=& 4 \bs x'^T [4 \bs \beta \bsmc J - 4 \bs \beta \bsmc J \bs B^{-1} \bsmc J - \bsmc J] \bs x + 8 \bs \alpha^T \bs \beta \bsmc J [1 - \bs B^{-1} \bsmc J] \bs x
\end{eqnarray}
for \(\bs \beta = -\frac{1}{2} \bsmc J^T (\bs B - \bsmc J \bs B^{-1} \bsmc J)^{-1} \).

This allows us to rewrite
\begin{equation}
  U U^*(\bs x, \bs x') = G_1(\bs A, \bs 0) \sum_{\bs x_3 \in (\mathbb Z/p \mathbb Z)^{2n}} \exp \left[ -\frac{\pi i}{p} \left( \bs x_3^T \bs A' \bs x_3 + 2 \bs v'^T \bs x_3 + 2 c' \right) \right].
\end{equation}

By Proposition~\ref{prop:Gausssums}b, since \(\bs A' = \bs 0\), the sum over \(\bs x_3\) is non-zero if and only if \(\bs v' = 0\):
\begin{eqnarray}
  \bs x^T &=& (2 \bs x' + \bs \alpha)^T \bsmc J (\bsmc J \bs B - 1)^{-1} \bsmc J - \bs x'^T \\
          &=& \bs x'^T \bsmc J (2 + (\bsmc J \bs B - 1))(\bsmc J \bs B - 1)^{-1} \bsmc J\nonumber\\
          &&+ \bs \alpha^T \frac{1}{2} \bsmc J(1 + \bsmc J \bs B + 1 - \bsmc J \bs B)(\bsmc J \bs B - 1)^{-1} \bsmc J \nonumber\\
          &=& - \bs x'^T \bsmc J (1 + \bsmc J \bs B)(1 - \bsmc J \bs B)^{-1} \bsmc J\nonumber\\
          &&- \bs \alpha^T \frac{1}{2} \bsmc J (1 + \bsmc J \bs B)(1 - \bsmc J \bs B)^{-1} \bsmc J + \frac{\bs \alpha^T}{2} \nonumber\\
          &=& (\bs x'^T + \frac{\bs \alpha^T}{2}) (- \bsmc J \bsmc M^{-1} \bsmc J) + \frac{\bs \alpha^T}{2} \nonumber\\
  \implies \bs x &=& (- \bsmc J \bsmc M^{-1} \bsmc J)^T \left(\bs x' + \frac{\bs \alpha}{2}\right) + \frac{\bs \alpha}{2} \nonumber\\
          &=& - \bsmc J (- \bsmc J \bsmc M \bsmc J) \bsmc J \left(\bs x' + \frac{\bs \alpha}{2}\right) + \frac{\bs \alpha}{2} \nonumber\\
          &=& \bsmc M \left(\bs x' + \frac{\bs \alpha}{2}\right) + \frac{\bs \alpha}{2}. \nonumber
\end{eqnarray}

Thus we find that the sum over \(\bs x_3\) is non-zero i.f.f. \(\bs x = \bsmc M \left(\bs x' + \frac{\bs \alpha}{2}\right) + \frac{\bs \alpha}{2}\) and for these values the phase \(2 c'\) is equal to zero too. We can see this last fact by examining the quadratic, linear and constant parts of \(2c'\) w.r.t. \(\bs x'\) separately.

Quadratic terms in \(\bs x'\):
\begin{eqnarray}
  && 4 \bs x'^T [2 \bsmc J (\bs B - \bsmc J \bs B^{-1} \bsmc J)^{-1} \bsmc J - 2 \bsmc J (\bs B - \bsmc J \bs B^{-1} \bsmc J)^{-1} \bsmc J \bs B^{-1} \bsmc J - \bsmc J] (1 + \bsmc J \bs B)^{-1}(1 - \bsmc J \bs B) \bs x' \\
  &=& 4 \bs x'^T \bsmc J \bs x' \nonumber\\
  &=& 0. \nonumber
\end{eqnarray}
Using the results from the same simplification from the work on the quadratic part to simplify the first two terms in the linear part, we find:
\begin{eqnarray}
  && 2 \bs x'^T \bsmc J \bs \alpha + 2 \bs x'^T \bsmc J (\bs B - \bsmc J \bs B^{-1} \bsmc J)^{-1} \bsmc J (1 - \bs B^{-1} \bsmc J) (1 + \bsmc J \bs B) \bs \alpha \\
  && + 4 \bs \alpha^T \bsmc J (\bs B - \bsmc J \bs B^{-1} \bsmc J)^{-1} \bsmc J (1 - \bs B^{-1} \bsmc J) (1 - \bsmc J \bs B) (1 + \bsmc J \bs B)^{-1} \bs x' \nonumber\\
  &=& 4 \bs x'^T [-(\bsmc J - \bsmc J \bs B \bsmc J)^{-1} + (\bsmc J - \bsmc J \bs B \bsmc J)^{-1}] \bs \alpha \nonumber\\
  &=& 0. \nonumber
\end{eqnarray}
Constant terms:
\begin{eqnarray}
  && 2 \bs \alpha^T \bsmc J (\bs B - \bsmc J \bs B^{-1} \bsmc J)^{-1} \bsmc J (1 - \bs B^{-1} \bsmc J)((1 + \bsmc J \bs B)^{-1}(1 - \bsmc J \bs B) + 1) \bs \alpha\\
  &=& -4 \bs \alpha^T (\bsmc J \bs B \bsmc J \bs B \bsmc J)^{-1} \bs \alpha \nonumber
\end{eqnarray}
Since \(((\bsmc J \bs B \bsmc J \bs B \bsmc J)^{-1})^T = -(\bsmc J \bs B \bsmc J \bs B \bsmc J)^{-1}\), it follows that
\begin{equation}
  -4 \bs \alpha^T (\bsmc J \bs B \bsmc J \bs B \bsmc J)^{-1} \bs \alpha = 4 \bs \alpha^T (\bsmc J \bs B \bsmc J \bs B \bsmc J)^{-1} \bs \alpha = 0.
\end{equation}
\end{widetext}

In particular, for these values of \(\bs x\) and \(\bs x'\),  \(U^* U(\bs x, \bs x') = d^{2n}\) since the phase \(2 c'\) is equal to zero and \(G_1(0) = d^{2n}\) for a \(4n\) dimensional summation variable \((\bs x_1, \bs x_2)\).

\section{First Non-Trivial Example}
\label{app:firstnontrivialex}

We verify that the operator with the action \(S_9\) corresponds to a unitary operator:
\begin{widetext}
\begin{eqnarray}
  \left(U_9 U_9^* \right) (\bs x) &=& \left(\frac{1}{3^2}\right)^2\sum_{\bs x', \bs x'' \in (\mathbb Z/3^2\mathbb Z)^2} U_9(\bs x'') U_9^*(\bs x') \exp\left(\frac{2 \pi i}{3^2} \Delta_3(\bs x, \bs x', \bs x'') \right)\\
                                        &=& \frac{1}{3^4}\sum_{\bs x', \bs x'' \in (\mathbb Z/3^2\mathbb Z)^2} \exp\left\{ \frac{2 \pi i}{3^2} \left[ C (-{x_q'}^3 + {x_q''}^3) + B (-{x_q'}^2 + {x_q''}^2) + \bs \alpha \bsmc J (\bs x' - \bs x'') + 2 \bs x^T \bsmc J (\bs x' - \bs x'') + 2 \bs x'^T \bsmc J \bs x'' \right] \right\} \nonumber\\
                                        &=& \frac{1}{3^4}\sum_{\bs x', \bs x'' \in (\mathbb Z/3^2\mathbb Z)^2} \exp\left\{ \frac{2 \pi i}{3^2} \left[ C (-{x_q'}^3 + {x_q''}^3) + B (-{x_q'}^2 + {x_q''}^2) - \alpha_p (x_q' - x_q'') - 2 x_p (x_q' - x_q'') \right] \right\} \nonumber\\
                                        && \qquad \qquad \qquad \times \exp\left\{ \frac{2 \pi i}{3^2} \left[ (-2 x_q - \alpha_q + 2 x'_q) x''_p + (\alpha_q + 2 x_q - 2 x_q'') x_p' \right] \right\} \nonumber\\
                                        &=& \sum_{x'_q, \bs x''_q \in \mathbb Z/3^2\mathbb Z} \exp\left\{ \frac{2 \pi i}{3^2} \left[ C (-{x_q'}^3 + {x_q''}^3) + B (-{x_q'}^2 + {x_q''}^2) - \alpha_p (x_q' - x_q'') - 2 x_p (x_q' - x_q'') \right] \right\} \nonumber\\
                                        && \qquad \qquad \qquad \times \delta\left[\left(2 x'_q - 2 x_q - \alpha_q\right)\Mod 3^2 \right] \delta\left[ \left(-2 x_q'' + 2 x_q + \alpha_q\right) \Mod 3^2 \right] \nonumber\\
                                        &=& \sum_{x'_q, \bs x''_q \in \mathbb Z/3^2\mathbb Z} \exp\left\{ \frac{2 \pi i}{3^2} \left[ C (-{x_q'}^3 + {x_q''}^3) + B (-{x_q'}^2 + {x_q''}^2) - \alpha_p (x_q' - x_q'') - 2 x_p (x_q' - x_q'') \right] \right\} \nonumber\\
                                        && \qquad \qquad \qquad \times \delta\left[ \left(x'_q - x''_q\right) \Mod 3^2 \right] \delta\left[ \left(2 x'_q + 2 x''_q - 4 x_q - 2 \alpha_q \right) \Mod 3^2 \right] \nonumber\\
                                        &=& 1. \nonumber
\end{eqnarray}
So we now consider
\begin{eqnarray}
  U_9 U_9^*(\bs x, \bs x') &=& \frac{1}{9^{2}} \sum_{\bs x_1, \bs x_2, \bs x_3 \in (\mathbb Z/3^2 \mathbb Z)^{2N}} U(\bs x_1) U^*(\bs x_2) \exp\left[ \frac{2 \pi i}{p} \Delta_5(\bs x_3, \bs x, \bs x_1, \bs x', \bs x_2)\right] \\
  &=& \frac{1}{9^2} \sum_{\bs x_1, \bs x_2, x_3 \in (\mathbb Z/3^{2} \mathbb Z)^2} \nonumber\\
  && \exp \left\{ \frac{2 \pi i}{3^2} \left[ S_9(\bs x_1) - S_9(\bs x_2) - \left( \begin{array}{c} \bs x_1\\ \bs x_2 \end{array}\right)^T\left( \begin{array}{cc} 0& \bsmc J\\ -\bsmc J& 0 \end{array} \right) \left( \begin{array}{c} \bs x_1\\ \bs x_2 \end{array}\right) + 2 \left(\begin{array}{c} \bs x_3 - \bs x + \bs x'\\ \bs x_3 - \bs x - \bs x' \end{array} \right)^T \bsmc J \left( \begin{array}{c} \bs x_1\\ \bs x_2 \end{array}\right) + c \right] \right\}, \nonumber
\end{eqnarray}
\end{widetext}
where we remind ourselves that
\begin{equation}
  c = 2 \left[\bs x'^T \bsmc J \bs x + \bs x_3^T \bsmc J (\bs x + \bs x') \right].
\end{equation}

We first sum away the linear monomials in \(x_{1p}\) and \(x_{2p}\) to produce Kronecker delta functions:
\begin{widetext}
\begin{eqnarray}
 U_9 U_9^*(\bs x, \bs x') &=& \frac{1}{{9^2}}\sum_{\bs x_1, \bs x_2, \bs x_3 \in (\mathbb Z/3^2 \mathbb Z)^2} \exp \left\{ \frac{2 \pi i}{3^2} \left[ C x_{1q}^3 + B x_{1q}^2 - C x_{2q}^3 - B x_{2q}^2 + c \right] \right\}, \\
  && \times \exp \left\{ \frac{2 \pi i}{3^2} \left[ \bs \alpha^T \bsmc J (\bs x_1 - \bs x_2) - \left(\begin{array}{c} \bs x_1\\ \bs x_2 \end{array}\right)^T\left( \begin{array}{cc} 0& \bsmc J\\ -\bsmc J& 0 \end{array} \right) \left( \begin{array}{c} \bs x_1\\ \bs x_2 \end{array}\right) + 2 \left(\begin{array}{c} \bs x_3 - \bs x + \bs x'\\ \bs x_3 - \bs x - \bs x' \end{array} \right)^T \bsmc J \left( \begin{array}{c} \bs x_1\\ \bs x_2 \end{array}\right) \right] \right\} \nonumber\\
  &=& \frac{1}{{9^2}} \sum_{\bs x_1, \bs x_2, \bs x_3 \in (\mathbb Z/3^2 \mathbb Z)^2} \exp \left\{ \frac{2 \pi i}{3^2} \left[ C x_{1q}^3 + B x_{1q}^2 - C x_{2q}^3 - B x_{2q}^2 + c \right] \right\}, \nonumber\\
  && \times \exp \left\{ \frac{2 \pi i}{3^2} \left[ \alpha_p (x_{1q} - x_{2q}) - 2(x_{3p} - x_p + x'_p) x_{1q} - 2(x_{3p} - x_p - x'_p) x_{2q} \right] \right\} \nonumber\\
  && \times \exp \left\{ \frac{2 \pi i}{3^2} \left[ - \alpha_q + 2( x_{2q} + x_{3q} - x_q + x'_q) \right] x_{1p} \right\} \exp \left\{ \frac{2 \pi i}{3^2} \left[ \alpha_q + 2( -x_{1q} + x_{3q} - x_q - x'_q) \right] x_{2p} \right\} \nonumber\\
  &=& \sum_{\substack{\bs x_{1q}, \bs x_{2q} \in \mathbb Z/3^2 \mathbb Z\\ \bs x_3 \in (\mathbb Z/3^2 \mathbb Z)^2}} \exp \left\{ \frac{2 \pi i}{3^2} \left[ C x_{1q}^3 + B x_{1q}^2 - C x_{2q}^3 - B x_{2q}^2 + c \right] \right\}, \nonumber\\
  && \times \exp \left\{ \frac{2 \pi i}{3^2} \left[ \alpha_p (x_{1q} - x_{2q}) - 2(x_{3p} - x_p + x'_p) x_{1q} - 2(x_{3p} - x_p - x'_p) x_{2q} \right] \right\} \nonumber\\
  && \times \delta\left\{ \left[ \alpha_q - 2( x_{2q} + x_{3q} - x_q + x'_q) \right] \Mod 3^2\right\} \delta\left\{ \left[-\alpha_q - 2( -x_{1q} + x_{3q} - x_q - x'_q) \right] \Mod 3^2 \right\}\nonumber\\
  &=& \sum_{\substack{\bs x_{1q}, \bs x_{2q} \in \mathbb Z/3^2 \mathbb Z\\ \bs x_3 \in (\mathbb Z/3^2 \mathbb Z)^2}} \exp \left\{ \frac{2 \pi i}{3^2} \left[ C x_{1q}^3 + B x_{1q}^2 - C x_{2q}^3 - B x_{2q}^2 + c \right] \right\}, \nonumber\\
  && \times \exp \left\{ \frac{2 \pi i}{3^2} \left[ \alpha_p (x_{1q} - x_{2q}) - 2(x_{3p} - x_p + x'_p) x_{1q} - 2(x_{3p} - x_p - x'_p) x_{2q} \right] \right\} \nonumber\\
  && \times \delta\left\{ \left[ -2( x_{2q} + 2 x_{3q} - 2 x_q - x_{1q}) \right] \Mod 3^2 \right\} \delta \left\{ \left[ 2 \left( \alpha_q - x_{2q} - x_{1q} - 2 x'_q \right) \right] \Mod 3^2 \right\} \nonumber
\end{eqnarray}

We replace the \(x_{3q}\) sum and sum away \(x_{3p}\):
\begin{eqnarray}
 U_9 U_9^*(\bs x, \bs x') &=& \sum_{\substack{x_{1q}, x_{2q} \in \mathbb Z/3^2 \mathbb Z\\ \bs x_3 \in (\mathbb Z/3^2 \mathbb Z)^2}} \exp \left\{ \frac{2 \pi i}{3^2} \left[ C x_{1q}^3 + B x_{1q}^2 - C x_{2q}^3 - B x_{2q}^2 + 2(\bs x'^T \bsmc J \bs x - x_{3p}(x_q + x'_q) ) \right] \right\}, \\
                          && \times \exp \left\{ \frac{2 \pi i}{3^2} \left[ \alpha_p (x_{1q} - x_{2q}) - 2(x_{3p} - x_p + x'_p) x_{1q} - 2(x_{3p} - x_p - x'_p) x_{2q} \right] \right\} \nonumber\\
                          && \times \exp \left\{ \frac{2 \pi i}{3^2} 2 (x_p + x'_p) x_{3q} \right\}  \delta\left\{ \left[ -2( x_{2q} + 2 x_{3q} - 2 x_q - x_{1q}) \right] \Mod 3^2 \right\} \delta \left\{ \left[ 2 \left( \alpha_q - x_{2q} - x_{1q} - 2 x'_q \right) \right] \Mod 3^2 \right\} \nonumber\\
                          &=& \sum_{x_{1q}, x_{2q}, x_{3p} \in \mathbb Z/3^2 \mathbb Z} \exp \left\{ \frac{2 \pi i}{3^2} \left[ C x_{1q}^3 + B x_{1q}^2 - C x_{2q}^3 - B x_{2q}^2 + 2(\bs x'^T \bsmc J \bs x - x_{3p}(x_q + x'_q) ) \right] \right\}, \nonumber\\
                          && \times \exp \left\{ \frac{2 \pi i}{3^2} \left[ \alpha_p (x_{1q} - x_{2q}) - 2(x_{3p} - x_p + x'_p) x_{1q} - 2(x_{3p} - x_p - x'_p) x_{2q} \right] \right\} \nonumber\\
                          && \times \exp \left\{ \frac{2 \pi i}{3^2} \left[ 2 x_q (x_p + x'_p) - (x_{2q}-x_{1q})(x_p + x'_p) \right] \right\} \delta \left\{ \left[ 2 \left( \alpha_q - x_{2q} - x_{1q} - 2 x'_q \right) \right] \Mod 3^2 \right\} \nonumber\\
                          &=& \sum_{x_{1q}, x_{2q}, x_{3p} \in \mathbb Z/3^2 \mathbb Z} \exp \left\{ \frac{2 \pi i}{3^2} \left[ C x_{1q}^3 + B x_{1q}^2 - C x_{2q}^3 - B x_{2q}^2 + 2(\bs x'^T \bsmc J \bs x) \right] \right\}, \nonumber\\
                          && \times \exp \left\{ \frac{2 \pi i}{3^2} \left[ \alpha_p (x_{1q} - x_{2q}) - 2(- x_p + x'_p) x_{1q} - 2(- x_p - x'_p) x_{2q} + 2 x_q (x_p + x'_p) - (x_{2q}-x_{1q})(x_p + x'_p) \right] \right\} \nonumber\\
                          && \times \exp \left\{ \frac{2 \pi i}{3^2} 2 \left[ - x_q - x'_q - x_{2q} - x_{1q}\right] x_{3p}\right\} \delta \left\{ \left[ 2 \left( \alpha_q - x_{2q} - x_{1q} - 2 x'_q \right) \right] \Mod 3^2 \right\} \nonumber\\
                          &=& 9 \sum_{x_{1q}, x_{2q} \in \mathbb Z/3^2 \mathbb Z} \exp \left\{ \frac{2 \pi i}{3^2} \left[ C x_{1q}^3 + B x_{1q}^2 - C x_{2q}^3 - B x_{2q}^2 + 2(\bs x'^T \bsmc J \bs x) \right] \right\}, \nonumber\\
                          && \times \exp \left\{ \frac{2 \pi i}{3^2} \left[ \alpha_p (x_{1q} - x_{2q}) - 2(- x_p + x'_p) x_{1q} - 2(- x_p - x'_p) x_{2q} + 2 x_q (x_p + x'_p) - (x_{2q}-x_{1q})(x_p + x'_p) \right] \right\} \nonumber\\
                          && \times \delta \left\{ \left[ 2 \left( - x_q - x'_q - x_{2q} - x_{1q}\right) \right] \Mod 3^2 \right\} \delta \left\{ \left[ 2 \left( \alpha_q - x_{2q} - x_{1q} - 2 x'_q \right) \right] \Mod 3^2 \right\} \nonumber\\
                          &=& 9 \sum_{x_{2q} \in \mathbb Z/3^2 \mathbb Z} \exp \left\{ \frac{2 \pi i}{3^2} \left[ C (-x_q - x'_q - x_{2q})^3 + B (-x_q - x'_q - x_{2q})^2 - C x_{2q}^3 - B x_{2q}^2 + 2(\bs x'^T \bsmc J \bs x) + 2 x_q (x_p + x'_p)\right] \right\}, \nonumber\\
                          && \times \exp \left\{ \frac{2 \pi i}{3^2} \left[ \alpha_p (-x_q - x'_q - 2 x_{2q}) - 2(-x_p + x'_p) (-x_q - x'_q - x_{2q}) - 2(- x_p - x'_p) x_{2q} - (2 x_{2q}+ x_q + x'_q)(x_p + x'_p) \right] \right\} \nonumber\\
                          && \times \delta \left\{ \left[ 2 \left( \alpha_q + x_q - x'_q \right) \right] \Mod 3^2 \right\}, \nonumber\\
                          &=& \begin{cases} 9 \sum_{x_{2q} \in \mathbb Z/3^2 \mathbb Z} \exp \left\{ \frac{2 \pi i}{3^2} \left[ C' x_{2q}^3 + B' x_{2q}^2 + \alpha' x_{2q} + c'\right] \right\} & \text{for}\, x_q = (- \alpha_q + x'_q) \Mod 3^2\\ 0 & \text{otherwise.}\end{cases} \nonumber
\end{eqnarray}
\end{widetext}
where
\begin{equation}
  C' = -2 C,
\end{equation}
\begin{equation}
  B' = -3 C (2 x'_q - \alpha_q),
\end{equation}
\begin{equation}
  \alpha' = (2 x'_p - 2 x_p - 2 \alpha_p + (2 x'_q - \alpha_q)(2 B - 6 C x'_q + 3 C \alpha_q)),
\end{equation}
and
\begin{equation}
  c' = -(2 x'_q - \alpha_q)(- x'_p + x_p + \alpha_p + (2 x'_q - \alpha_q)(-B + C(2 x'_q - \alpha_q)).
\end{equation}
where in the last step we could have replaced \(x_{2q}\) instead of \(x_{1q}\) with similar results.

This leaves a single sum over \(x_{2q}\) of an exponentiated cubic polynomial \(S(x_{2q}) = C' x_{2q}^3 + B' x_{2q}^2 + \alpha' x_{2q} + c'\) with coefficients in \(\mathbb Z\). We proceed to find the exponential argument's zeros \(\{\tilde x_{2q}\}\), the stationary phase points:
\begin{equation}
  \left.\partder{S}{x_{2q}}\right|_{x_{2q} = \tilde x_{2q}} = 0 = 3 C' \tilde x^2_{2q} + 2 B' \tilde x_{2q} + \alpha',
\end{equation}
\begin{equation}
  \Leftrightarrow \tilde x_{2q} = \left[- 2 B' \pm \left( 4 B'^2 - 12 C' \alpha' \right)^{1/2} \right] (6 C')^{-1},
\end{equation}
where all the arithmetic operations are understood to be taken \(\Mod 3^2\).

Therefore, by Theorem~\ref{th:statphase}(a) and (b),
\begin{eqnarray}
  &&U_9U_9^*(\bs x, \bs x')\\
  &=& 9 \sum_{\bar x_{2q} \in \mathbb Z/ 3 \mathbb Z} \left\{ \sum_{\substack{x_{2q} \in \mathbb Z/ 3^2 \mathbb Z\\ x_{2q} \Mod 3 = \bar x_{2q}}} \exp \left[ \frac{2 \pi i}{3^2} S(x_{2q})\right] \right\} \nonumber\\
  &=& 9 \sum_{\{\tilde x_{2q}\} \Mod 3} 3 \exp\left[\frac{2 \pi i}{3^2} S(\tilde x_{2q}) \right] G_0(H_{\tilde x_{2q}},3^{-1}\nabla S(\tilde x_{2q}) ), \nonumber\\
  &=& 9 \sum_{\{\tilde x_{2q}\} \Mod 3} 3 \exp\left[\frac{2 \pi i}{3^2} S(\tilde x_{2q}) \right], \nonumber
\end{eqnarray}
where we remind ourselves that \(H_{\tilde x_{2q}} = \left. \frac{\partial^2 S}{\partial x_\mu \partial x_\nu} \right|_{x = \tilde x_{2q}}\).

For instance, for \(C=-2\), \(B=-4\) and \(\bs \alpha = (0,4)\),
\begin{eqnarray}
  \partder{S(x_{2q})}{x_{2q}} &=& 12 x_{2q}^2 + 12 (-4+2 x'_q) x_{2q} - 2 x_p + 2 x'_p\\
                                  && + (-4 + 2 x'_q)(-32 + 12 x'_q), \nonumber
\end{eqnarray}
for \(x_q = - \alpha_q + x'_q \Mod 3^2\), and so for \(\bs x = (0,0)\) and \(\bs x' = (1,4)\),
\begin{equation}
  \partder{S(x_{2q})}{x_{2q}} \Mod 3 = \left[66 + 12 x_{2q} (4 + x_{2q})\right] \Mod 3 = 0 \, \forall\, x_{2q}.
\end{equation}
Therefore, \(\tilde x_{2q}\) takes all values \(\Mod 3\) and so
\begin{eqnarray}
  &&U_9U_9^*((0,0), (1,4))\\
  &=& 9 \left\{ \exp\left[\frac{2 \pi i}{3^2} S(0) \right] \right. \nonumber\\
      && \left.+ \exp\left[\frac{2 \pi i}{3^2} S(1) \right] + \exp\left[\frac{2 \pi i}{3^2} S(2) \right]\right\}, \nonumber\\
  &=& 9 \left\{ \exp\left[\frac{2 \pi i}{3^2} 5 \right] \right. \nonumber\\
  && \left.+ \exp\left[\frac{2 \pi i}{3^2} 0 \right] + \exp\left[\frac{2 \pi i}{3^2} 4 \right]\right\}, \nonumber
\end{eqnarray}

On the other hand, for \(\bs x = (0,0)\) and \(\bs x' = (2,4)\),
\begin{equation}
  \partder{S(x_{2q})}{x_{2q}} \Mod 3 = \left[68 + 12 x_{2q} (4 + x_{2q})\right] \Mod 3 \ne 0 \, \forall\, x_{2q},
\end{equation}
and so \(U_9U_9^*((0,0), (2,4)) = 0\).

This example illustrates the usefulness of this periodized stationary phase method for evaluating non-Clifford gate Weyl symbols and its potential for simplifying or reducing the full sum.

\section{Comparison between the Discrete Periodized Stationary Phase Method and the Stationary Phase Approximation in the Continuous Case}
\label{sec:periodizedstatphase}

The stationary phase method in the continuous case can be described as representing a integral over a continuous domain by a set of discrete points \(\{\tilde x_i\}\). These are the stationary points of the phase of the integrand. The phase is expanded by Taylor's theorem at these discrete points \(\{\tilde x_i\}\), and if the expansion is truncated at quadratic order, then Gaussian integrals result.

The notion of the stationary phase method differs when there is no longer a proper Euclidean metric to define continuous distances or areas for the integrand's domain. Similarly, the traditional version of Taylor's theorem does not hold. Instead, we must use the ``periodized'' version of Taylor's theorem~\cite{Bourbaki90} in a ``periodized'' stationary phase approximation. In this version, the points \(\{\tilde x_i\}\) where the phase's derivative is zero correspond to critical points where the ``periodized'' phase is stationary, i.e., where the phase, \emph{taken every \(x_i\) points}, slows down so that it ceases cancelling its opposing contributions around the unit circle with equal weights.

In the discrete case we consequently look for stationarity along equally-spaced periodic intervals that span the whole summand, instead of at particular points of the integrand; stationary phase points correspond to ``reduced'' points representing periodic intervals of the summand at which an expansion to second order in the action is made producing quadratic Gauss sum contributions.

\section{Qutrit \(\pi/8\) Gate}
\label{app:qutritpi8gate}

\(S_{\pi/8}(\bs x) \notin \mathbb Z[\bs x]\) but it does lie in \(\mathbb Z[\bs x]\) when the domain is increased from \(3\) to \(3^2\); we want to consider the sum over \(x_q\) of \(U_{\pi/8}(\bs x) = U_{\pi/8}(x_q)\) and reexpress it in terms of an exponential such that it is a sum over the larger space \(\mathbb Z/ 3^2 \mathbb Z\) of a polynomial with integer coefficients. Hence, we want to find \(\alpha, \beta \in \mathbb Z\) s.t.
\begin{widetext}
  \begin{equation}
    \alpha x_q^{2 \times 3} + \beta x_q^{3} \Mod 3^2 = -2 \left(3 (x_q \Mod 3)^2 + (x_q \Mod 3)\right) \Mod 3^2.
  \end{equation}
\end{widetext}
  Substituting in \(x_q=1\) and \(x_q=3^{2}-1 = -1 (\Mod 3^{2})\) we find the system of equations
  \begin{equation}
    4 \times (-2) \Mod 3^2 = \alpha + \beta \Mod 3^2,
  \end{equation}
  and
  \begin{equation}
    14 \times (-2) \Mod 3^2 = \alpha - \beta \Mod 3^2,
  \end{equation}
  which has the solution \(\alpha = 0\) and \(\beta = 1\).
  
Therefore,
\begin{eqnarray}
  && \sum_{x_q \in \mathbb Z/3 \mathbb Z} \exp\left[ \frac{2 \pi i}{3} \left(-\frac{2}{3}\right) (x_q + 3 x_q^2) \right] \\
  &=& \frac{1}{3} \sum_{x_q \in \mathbb Z/3^2 \mathbb Z} \exp \left[ \frac{2 \pi i}{3^2} x_q^{3} \right] \nonumber\\
  &\equiv& \frac{1}{3} \sum_{x_q \in \mathbb Z/3^2 \mathbb Z} \exp \left[ \frac{2 \pi i}{3^2} S'_{\pi/8}(\bs x) \right]. \nonumber
\end{eqnarray}
Notice that the first line consists of a sum of an exponentiated polynomial with coefficients in \(\mathbb Q\) while the second line is a sum over a larger space, but this time of an exponentiated polynomial with coefficients in \(\mathbb Z\). Thus, we have accomplished what we set out to do.

Since we will not be performing the sum above in isolation but with other exponentiated terms, so it is important to establish a slightly more general result:
\begin{eqnarray}
  &&\sum_{x_q \in \mathbb Z/3 \mathbb Z} \exp\left\{ \frac{2 \pi i}{3} \left[ -\frac{2}{3} (x_q + 3 x_q^2) + f(\bs x) \right] \right\} \\
  &=& \frac{1}{3} \sum_{x_q \in \mathbb Z/3^2 \mathbb Z} \exp \left\{ \frac{2 \pi i}{3^2} \left[ S'_{\pi/8}(\bs x) + 3 f(\bs x) \right] \right\}, \nonumber
\end{eqnarray}
where \(f(\bs x) \in \mathbb Z[x_1, \ldots, x_n]\).

\begin{widetext}
 We now consider
\begin{eqnarray}
  U_{\pi/8} U^*_{\pi/8}(\bs x, \bs x') &=& \frac{1}{3^2} \sum_{\bs x_1, \bs x_2, \bs x_3 \in (\mathbb Z/3 \mathbb Z)^2} U_{\pi/8}(\bs x_1) U_{\pi/8}^*(\bs x_2) \exp\left[ \frac{2 \pi i}{3} \Delta_5(\bs x_3, \bs x, \bs x_1, \bs x', \bs x_2)\right], \\
                                       &=& \frac{1}{3^2} \frac{1}{3^{2}} \sum_{\substack{\bs x_3 \in (\mathbb Z/3 \mathbb Z)^2\\x_{1p}, x_{2p} \in \mathbb Z/3 \mathbb Z, \, x_{1q}, x_{2q} \in \mathbb Z/ 3^{2} \mathbb Z}} \exp \left\{ \frac{2 \pi i}{3^{2}} \left[ S'_{\pi/8}(\bs x_1) - S'_{\pi/8}(\bs x_2) \right] \right\} \nonumber\\
  && \times \exp \left\{ \frac{2 \pi i}{3^2} \left[- 3 \left( \begin{array}{c} \bs x_1\\ \bs x_2 \end{array}\right)^T\left( \begin{array}{cc} 0& \bsmc J\\ -\bsmc J& 0 \end{array} \right) \left( \begin{array}{c} \bs x_1\\ \bs x_2 \end{array}\right) - 2 \times 3 \left(\begin{array}{c} \bs x_3 - \bs x + \bs x'\\ \bs x_3 - \bs x - \bs x' \end{array} \right)^T \bsmc J \left( \begin{array}{c} \bs x_1\\ \bs x_2 \end{array}\right) + 3 c \right] \right\}, \nonumber
\end{eqnarray}
\end{widetext}
where we remind ourselves again that
\begin{equation}
  c = 2 \left[\bs x'^T \bsmc J \bs x + \bs x_3^T \bsmc J (\bs x + \bs x') \right].
\end{equation}

Simplifying, we find:
\begin{widetext}
\begin{eqnarray}
U_{\pi/8}U^*_{\pi/8}(\bs x, \bs x') &=& \frac{1}{3} \sum_{x_{2q} \in \mathbb Z/ 3^2 \mathbb Z} \exp \left\{ \frac{2 \pi i}{3^2} \left[ (-x_{2q} + x_q + x'_q)^{3}  - x_{2q}^{3} \right] \right\} \\
                                      && \times \exp \left\{ \frac{2 \pi i}{3^2} 3 \left[ 2 (x_p - x'_p) x_{2q} + (x_p + 3 x'_p) x'_q - (3 x_p + x'_p) x_q  \right] \right\} \nonumber\\
  \label{eq:qutritpi8gate_appendix}
                      && \times \delta \left[ 2 \times 3 \left( x_q - x'_q \right) \Mod 3^2 \right] \nonumber\\
                      &\equiv& \begin{cases}\frac{1}{3} \sum_{x_{2q} \in \mathbb Z/ 3^2 \mathbb Z} \exp \left[ \frac{2 \pi i}{3^2} S''_{\pi/8}(x_{2q}, x_q, x_p, x'_p) \right] & \text{if}\,\, x'_q = x_q \,(\hskip-3pt \Mod 3),\\ 0 & \text{otherwise} \end{cases}. \nonumber
\end{eqnarray}

Therefore, by Theorem~\ref{th:statphase}(a) and (b), for \(x'_q = x_q \,(\hskip-3pt \Mod 3)\) (and \(j=1\)),
\begin{eqnarray}
  \label{eq:qutritpi8gatestatphase_appendix}
  && U_{\pi/8}U_{\pi/8}^*(\bs x, \bs x') \\
  &=& \frac{1}{3} \sum_{\bar x_{2q} \in \mathbb Z/ 3 \mathbb Z} \left\{ \sum_{\substack{x_{2q} \in \mathbb Z/ 3^2 \mathbb Z\\ x_{2q} \Mod 3 = \bar x_{2q}}} \exp \left[ \frac{2 \pi i}{3^2} S''_{\pi/8}(x_{2q})\right] \right\} \nonumber\\
  &=& \frac{1}{3} \sum_{\{\tilde x_{2q}\} \Mod 3} \exp\left[\frac{2 \pi i}{3^2} S''_{\pi/8}(\tilde x_{2q}) \right] G_0(H_{\tilde x_{2q}},3^{-1}\nabla S''_{\pi/8}(\tilde x_{2q}) ), \nonumber
\end{eqnarray}
where \(G_0\) is defined in Eq.~\ref{eq:Gausssum}.
\end{widetext}

We note that Eq.~\ref{eq:qutritpi8gate_appendix} is \(3\)-periodic in all of its arguments, and so its simplification by Theorem~\ref{th:statphase} is a special case: its sums can simply be restricted to be over \(\mathbb Z/ 3 \mathbb Z\) with the appropriate power of \(3\) added to compensate. It is easy to find that Eq.~\ref{eq:qutritpi8gatestatphase_appendix} is precisely this equation, since the preceding phase \(\exp\left[\frac{2 \pi i}{3^2} S''_{\pi/8}(\tilde x_{2q}) \right]\) is the only non-trivial part when \(G_0 = 1\) and the domain of summation is appropriately reduced.

As an example, consider \(x_q = x'_q = 0\) and \(x_p = x'_p = 1\):
\begin{equation}
\partder{S''_{\pi/8}(x_{2q},0,1,1,1)}{x_{2q}} = 6 x_{2q}^2,
\end{equation}
and we consider the critical points \(\tilde x_{2q}\) when the above is equal to \(0 \, (\hskip-3pt\Mod 3)\).
Hence, \(\tilde x_{2q} = \{0,1,2\}\) and so
\begin{eqnarray}
  &&U_{\pi/8}U_{\pi/8}^*(1,0,1,0) \\
  &=& \frac{1}{3} 3 \left\{ \exp\left[ \frac{2 \pi i}{3^2} S''(0) \right] G_0\left[H_0, 3^{-1} \nabla S''(0)\right] \right.\nonumber\\
  && \qquad \qquad \left. + \exp\left[ \frac{2 \pi i}{3^2} S''(1) \right] G_0\left[H_1, 3^{-1} \nabla S''(1)\right] \right. \nonumber\\
  && \qquad \qquad \left.+ \exp\left[ \frac{2 \pi i}{3^2} S''(2) \right] G_0\left[H_2, 3^{-1} \nabla S''(2)\right] \right\} \nonumber\\
  &=& \left\{ \exp\left[ \frac{2 \pi i}{3^2} 0 \right]  + \exp\left[ \frac{2 \pi i}{3^2} 2 \right] \right. \nonumber\\
  && \qquad \quad \left. + \exp\left[ \frac{2 \pi i}{3^2} 16 \right] \right\}. \nonumber
\end{eqnarray}
The Gaussians above all are equivalent to \(G_0(0,0) = 1\). Therefore, as we can see, we are really just performing the same sum as in Eq.~\ref{eq:qutritpi8gate_appendix} but over the smaller domain \(\mathbb Z/3\mathbb Z\) and compensating by multiplying in the correct powers of \(3\).

\begin{widetext}
We are interested in the magic state of this gate, which corresponds to it acting on \(\hat H \ket 0 = \ket{p=0}\). The Wigner function of \(\ket{p=0}\) is \(\rho'(\bs x) \equiv \frac{1}{3} \delta_{x_p, \bs 0}\). Hence,
\begin{eqnarray}
  U_{\pi/8} U^*_{\pi/8}(\bs x, \bs x') &=& \frac{1}{3^2} \sum_{\bs x_1, \bs x_2, \bs x_3 \in (\mathbb Z/3 \mathbb Z)^2} U_{\pi/8}(\bs x_1) U_{\pi/8}^*(\bs x_2) \exp\left[ \frac{2 \pi i}{3} \Delta_5(\bs x_3, \bs x, \bs x_1, \bs x', \bs x_2)\right], \\
                                       &=& \frac{1}{3^2} \sum_{\bs x_1, \bs x_2, \bs x_3 \in (\mathbb Z/3 \mathbb Z)^2} \exp \left[ \frac{2 \pi i}{3} S_{\pi/8}(\bs x_1) \right] \exp \left[- \frac{2 \pi i}{3} S_{\pi/8}(\bs x_2) \right] \nonumber\\
                                       && \times \exp \left\{ \frac{2 \pi i}{3} \left[ - \left( \begin{array}{c} \bs x_1\\ \bs x_2 \end{array}\right)^T\left( \begin{array}{cc} 0& \bsmc J\\ -\bsmc J& 0 \end{array} \right) \left( \begin{array}{c} \bs x_1\\ \bs x_2 \end{array}\right) - 2 \left(\begin{array}{c} \bs x_3 - \bs x + \bs x'\\ \bs x_3 - \bs x - \bs x' \end{array} \right)^T \bsmc J \left( \begin{array}{c} \bs x_1\\ \bs x_2 \end{array}\right) + c \right] \right\}, \nonumber\\
                                       &=& \frac{1}{3^2} \frac{1}{3^{2}} \sum_{\substack{\bs x_3 \in (\mathbb Z/3 \mathbb Z)^2\\x_{1p}, x_{2p} \in \mathbb Z/3 \mathbb Z, \, x_{1q}, x_{2q} \in \mathbb Z/ 3^{2} \mathbb Z}} \exp \left\{ \frac{2 \pi i}{3^{2}} \left[ S'_{\pi/8}(\bs x_1) - S'_{\pi/8}(\bs x_2) \right] \right\} \nonumber\\
  && \times \exp \left\{ \frac{2 \pi i}{3^2} \left[- 3 \left( \begin{array}{c} \bs x_1\\ \bs x_2 \end{array}\right)^T\left( \begin{array}{cc} 0& \bsmc J\\ -\bsmc J& 0 \end{array} \right) \left( \begin{array}{c} \bs x_1\\ \bs x_2 \end{array}\right) - 2 \times 3 \left(\begin{array}{c} \bs x_3 - \bs x + \bs x'\\ \bs x_3 - \bs x - \bs x' \end{array} \right)^T \bsmc J \left( \begin{array}{c} \bs x_1\\ \bs x_2 \end{array}\right) + 3 c \right] \right\}, \nonumber
\end{eqnarray}
where we remind ourselves again that
\begin{equation}
  c = 2 \left[\bs x'^T \bsmc J \bs x + \bs x_3^T \bsmc J (\bs x + \bs x') \right].
\end{equation}

We first sum away the linear monomials in \(x_{1p}\) and \(x_{2p}\) to produce Kronecker delta functions:
\begin{eqnarray}
 U_{\pi/8}U^*_{\pi/8}(\bs x, \bs x') &=& \frac{1}{3^4} \sum_{\substack{\bs x_3 \in (\mathbb Z/3 \mathbb Z)^2\\x_{1p}, x_{2p} \in \mathbb Z/3 \mathbb Z, \, x_{1q}, x_{2q} \in \mathbb Z/ 3^{2} \mathbb Z}} \exp \left\{ \frac{2 \pi i}{3^{2}} \left[ S'_{\pi/8}(\bs x_{1q}) - S'_{\pi/8}(\bs x_{2q}) + 3 c \right] \right\}\\
  && \times \exp \left\{ \frac{2 \pi i}{3^{2}} \left[ 2 \times 3(x_{3p} - x_p + x'_p) x_{1q} + 2 \times 3(x_{3p} - x_p - x'_p) x_{2q} \right] \right\} \nonumber\\
                     && \times \exp \left\{ \frac{2 \pi i}{3^{2}} \left[ -2 \times 3( -x_{2q} + x_{3q} - x_q + x'_q) \right] x_{1p} \right\} \exp \left\{ \frac{2 \pi i}{3^{2}} \left[ -2 \times 3( x_{1q} + x_{3q} - x_q - x'_q) \right] x_{2p} \right\} \nonumber\\
  &=& \frac{1}{3^2} \sum_{\substack{\bs x_3 \in (\mathbb Z/3 \mathbb Z)^2\\x_{1q}, x_{2q} \in \mathbb Z/ 3^{2} \mathbb Z}} \exp \left\{ \frac{2 \pi i}{3^{2}} \left[ S'_{\pi/8}(\bs x_{1q}) - S'_{\pi/8}(\bs x_{2q}) + 3 c \right] \right\} \nonumber\\
  && \times \exp \left\{ \frac{2 \pi i}{3^{4}} \left[ 2 \times 3(x_{3p} - x_p + x'_p) x_{1q} + 2 3(x_{3p} - x_p - x'_p) x_{2q} \right] \right\} \nonumber\\
  && \times \delta \left[ - 2 \times 3 ( -x_{2q} + x_{3q} - x_q + x'_q) \Mod 3^{2}\right] \delta\left[ -2 \times 3 ( x_{1q} + x_{3q} - x_q - x'_q) \Mod 3^2 \right]\nonumber\\
  &=& \frac{1}{3^2} \sum_{\substack{\bs x_3 \in (\mathbb Z/3 \mathbb Z)^2\\x_{1q}, x_{2q} \in \mathbb Z/ 3^{2} \mathbb Z}} \exp \left\{ \frac{2 \pi i}{3^{2}} \left[ S'_{\pi/8}(\bs x_{1q}) - S'_{\pi/8}(\bs x_{2q}) + 3 c \right] \right\} \nonumber\\
  && \times \exp \left\{ \frac{2 \pi i}{3^{4}} \left[ 2 \times 3 (x_{3p} - x_p + x'_p) x_{1q} + 2 \times 3(x_{3p} - x_p - x'_p) x_{2q} \right] \right\} \nonumber\\
  && \times \delta\left[ -2 \times 3 ( -x_{2q} + 2 x_{3q} - 2 x'_q + x_{1q}) \Mod 3^{2} \right] \delta \left[ -2 \times 3 \left( -x_{2q} - x_{1q} + 2 x_q \right) \Mod 3^{2} \right] \nonumber
\end{eqnarray}

We evaluate the \(x_{3q}\) sum and, as sum away \(x_{3p}\):
\begin{eqnarray}
  U_{\pi/8}U_{\pi/8}^*(\bs x, \bs x') &=& \frac{1}{3^2} \sum_{\substack{\bs x_3 \in (\mathbb Z/3 \mathbb Z)^2\\x_{1q}, x_{2q} \in \mathbb Z/ 3^{2} \mathbb Z}} \exp \left\{ \frac{2 \pi i}{3^{2}} \left[ S'_{\pi/8}(\bs x_{1q}) - S'_{\pi/8}(\bs x_{2q}) + 2 \times 3 \bs x'^T \bsmc J \bs x \right] \right\} \nonumber \\
                      && \times \exp \left\{ \frac{2 \pi i}{3^{2}} \left[ 2 \times 3 (x_{3p} - x_p + x'_p) x_{1q} + 2 \times 3 (x_{3p} - x_p - x'_p) x_{2q} - 2 \times 3 (x_q + x'_q) x_{3p} \right] \right\} \\
                      && \times \exp \left\{ \frac{2 \pi i}{3^{2}} \left[ 2 \times 3 (x_p + x'_p) x_{3q} \right] \right\} \nonumber \\
                      && \times \delta\left[ -2 \times 3 ( -x_{2q} + 2 x_{3q} - 2 x'_q + x_{1q}) \Mod 3^{2} \right] \delta \left[ -2 \times 3 \left( -x_{2q} - x_{1q} + 2 x_q \right) \Mod 3^{2} \right] \nonumber\\
                      &=& \frac{1}{3^2} \sum_{\substack{\bs x_{3p} \in \mathbb Z/3 \mathbb Z\\x_{1q}, x_{2q} \in \mathbb Z/ 3^{2} \mathbb Z}} \exp \left\{ \frac{2 \pi i}{3^{2}} \left[ S'_{\pi/8}(\bs x_{1q}) - S'_{\pi/8}(\bs x_{2q}) + 2 \times 3 \bs x'^T \bsmc J \bs x \right] \right\} \nonumber \\
                      && \times \exp \left\{ \frac{2 \pi i}{3^{2}} \left[ 2 \times 3 (-x_p + x'_p) x_{1q} + 2 \times 3 (- x_p - x'_p) x_{2q} + 3 (x_p + x'_p) (-x_{1q} + x_{2q}) + 2 p^{2m} (x_p + x'_p) x'_q \right] \right\} \nonumber \\
                      && \times \exp \left\{ \frac{2 \pi i}{3^{2}} \left[ - 2 \times 3 (-x_{1q} - x_{2q} + x_q + x'_q) x_{3p} \right] \right\} \delta \left[ -2 \times 3 \left( -x_{2q} - x_{1q} + 2 x_q \right) \Mod 3^{2} \right] \nonumber \\
                      &=& \frac{1}{3} \sum_{x_{1q}, x_{2q} \in \mathbb Z/ 3^{2} \mathbb Z} \exp \left\{ \frac{2 \pi i}{3^{2}} \left[ S'_{\pi/8}(\bs x_{1q}) - S'_{\pi/8}(\bs x_{2q}) + 2 \times 3 \bs x'^T \bsmc J \bs x \right] \right\} \nonumber \\
                      && \times \exp \left\{ \frac{2 \pi i}{3^{2}} \left[ 2 \times 3 (-x_p + x'_p) x_{1q} + 2 \times 3 (- x_p - x'_p) x_{2q} + 3 (x_p + x'_p) (-x_{1q} + x_{2q}) + 2 \times 3 (x_p + x'_p) x'_q \right] \right\} \nonumber \\
                      && \delta \left[ 2 \times 3 (x_{1q} + x_{2q} - x_q - x'_q) \Mod 3^{2} \right] \delta \left[ -2 \times 3 \left( -x_{2q} - x_{1q} + 2 x_q \right) \Mod 3^{2} \right] \nonumber \\
                      &=& \frac{1}{3} \sum_{x_{2q} \in \mathbb Z/ 3^2 \mathbb Z} \nonumber\\
                      && \times \exp \left\{ \frac{2 \pi i}{3^2} \left[ S'_{\pi/8}(-x_{2q} + x_q + x'_q) - S'_{\pi/8}(x_{2q}) + 2 \times 3 \bs x'^T \bsmc J \bs x + 2 \times 3 (x_p + x'_p) x'_q \right] \right\} \nonumber \\
                      && \times \exp \left\{ \frac{2 \pi i}{3^2} \left[ 2 \times 3(-x_p + x'_p) (-x_{2q} + x_q + x'_q) + 2 \times 3(- x_p - x'_p) x_{2q} + 3 (x_p + x'_p) (- x_q - x'_q  + 2 x_{2q}) \right] \right\} \nonumber\\ 
                      && \delta \left[ 2 \times 3 \left( x_q - x'_q \right) \Mod 3^2 \right] \nonumber\\
                      &=& \frac{1}{3} \sum_{x_{2q} \in \mathbb Z/ 3^2 \mathbb Z} \nonumber\\
                      && \times \exp \left\{ \frac{2 \pi i}{3^2} \left[ (-x_{2q} + x_q + x'_q)^{3}  \right] \right\} \nonumber\\
                      && \times \exp \left\{ \frac{2 \pi i}{3^2} \left[ - x_{2q}^{3} \right] \right\} \nonumber \\
                                      && \times \exp \left\{ \frac{2 \pi i}{3^2} 3 \left[ 2 (x_p - x'_p) x_{2q} + (x_p + 3 x'_p) x'_q - (3 x_p + x'_p) x_q  \right] \right\} \nonumber\\
                      && \times \delta \left[ 2 \times 3 \left( x_q - x'_q \right) \Mod 3^2 \right] \nonumber\\
                      &\equiv& \begin{cases}\frac{1}{3} \sum_{x_{2q} \in \mathbb Z/ 3^2 \mathbb Z} \exp \left[ \frac{2 \pi i}{3^2} S''_{\pi/8}(x_{2q}, x_q, x_p, x'_p) \right] & \text{if}\,\, x'_q = x_q \,(\hskip-3pt \Mod 3),\\ 0 & \text{otherwise} \end{cases}. \nonumber
\end{eqnarray}
\end{widetext}

We are interested in the magic state of this gate, which corresponds to it acting on \(\hat H \ket 0 = \ket{p=0}\). The Wigner function of \(\ket{p=0}\) is \(\rho'(\bs x) \equiv \frac{1}{3} \delta_{x_p, \bs 0}\). Hence,
\begin{widetext}
  \begin{eqnarray}
    \rho_{\pi/8}(\bs x) \equiv \sum_{\bs x'} U_{\pi/8}U^*_{\pi/8}(\bs x, \bs x') \rho'(\bs x') &=& \begin{cases} \frac{1}{3^2} \sum_{x_{2q} \in \mathbb Z/ 3^2 \mathbb Z} \exp \left[ \frac{2 \pi i}{3^2} S''_{\pi/8}(x_{2q}, x_q, x_p, x'_p{=}0) \right] & \text{if}\,\, x'_q = x_q \,(\hskip-3pt \Mod 3),\\ 0 & \text{otherwise} \end{cases}. \nonumber\\
                      &=& \frac{1}{3^2} \sum_{\substack{x_{2q} \in \mathbb Z/ 3^2 \mathbb Z\\x'_q \in \mathbb Z/ 3 \mathbb Z}} \nonumber\\
                      && \times \exp \left\{ \frac{2 \pi i}{3^2} \left[ (-x_{2q} + x_q + x'_q)^{3} \right] \right\} \nonumber\\
                      && \times \exp \left\{ \frac{2 \pi i}{3^2} \left[ -x_{2q}^{3} \right] \right\} \nonumber \\
                      && \times \exp \left\{ \frac{2 \pi i}{3^2} 3 \left[ (2 x_{2q} + x'_q - 3 x_q) x_p  \right] \right\} \delta \left[ 2 \times 3 \left( x_q - x'_q \right) \Mod 3^2 \right] \nonumber\\
                      &=& \frac{1}{3^2} \sum_{x_{2q} \in \mathbb Z/ 3^2 \mathbb Z} \nonumber\\
                                                                                               && \times \exp \left\{ \frac{2 \pi i}{3^2} \left[ (-x_{2q} + 2 x_q )^{3} \right] \right\} \nonumber\\
                      && \times \exp \left\{ \frac{2 \pi i}{3^2} \left[ -x_{2q}^{3} \right] \right\} \\
    \label{eq:weylpi8gatequtrit_appendix}
                      && \times \exp \left\{ \frac{2 \pi i}{3^2} 2 \times 3 \left[ (x_{2q} - x_q) x_p  \right] \right\}. \nonumber
  \end{eqnarray}
\end{widetext}

If we let \(S(x_p, x_q)\) be the phase, we note that \(\partder{S}{x_p} = \partder{S}{x_q} = 0 \Mod 3\) and \(\nabla^2 S \equiv H = 0 \Mod 3^0\) \(\forall \, x_p, \,x_q\). Hence, evaluation at any phase space point, or linear combination thereof, requires summation over all three reduced phase space points \(\tilde x_{2q} \in \mathbb Z/ 3 \mathbb Z\) since they are all critical points.

Again, we note that Eq.~\ref{eq:weylpi8gatequtrit_appendix} is \(3\)-periodic in all of its arguments, and so its simplification by Theorem~\ref{th:statphase} is again a particularly simple special case where its sums can simply be restricted to be over \(\mathbb Z/ 3 \mathbb Z\) with the appropriate power of \(3\) added to compensate.

\end{document}